\documentclass[letterpaper,UKenglish,cleveref, autoref, thm-restate,runningheads]{llncs}
\usepackage[T1]{fontenc}
\usepackage{graphicx}
\usepackage{amssymb}
\usepackage{amsmath}
\usepackage{thmtools}
\usepackage{hyperref}
\usepackage{xcolor}
\usepackage{enumerate}
\usepackage{listings}
\usepackage{microtype}
\usepackage{stackengine}
\usepackage{comment}
\usepackage{array}
\usepackage{todonotes}
\usepackage[table,xcdraw]{xcolor}
\usepackage{subcaption}
\usepackage{orcidlink}

\usepackage{apptools}
\newcommand{\restateref}[1]{\IfAppendix{\hyperref[#1]{$\star$}}{\hyperref[#1*]{$\star$}}}

\spnewtheorem{ourclaim}{Claim}{\itshape}{\rmfamily}

\newcommand{\redmst}{T_R}
\newcommand{\bluemst}{T_B}
\newcommand{\conv}{{\rm conv}}

\title{How many times can\\ two minimum spanning trees cross? \thanks{The full version of this paper is available on arXiv~\cite{OurArxiv}.}}

\author{Todor {Anti\'c}\inst{1}\orcidlink{0009-0008-6521-7987} \and
Morteza Saghafian\inst{2}\orcidlink{0000-0002-4201-5775} \and
Maria Saumell\inst{3}\orcidlink{0000-0002-4704-2609}\and
Felix Schr\"oder\inst{1}\orcidlink{0000-0001-8563-3517}\and
Josef Tkadlec\inst{4}\orcidlink{0000-0002-1097-9684}\and
Pavel Valtr\inst{1}\orcidlink{0000-0002-3102-4166}}

\institute{Department of Applied Mathematics, Faculty of Mathematics and Physics, Charles University, Czech Republic\\
\email{\{todor,schroder,valtr\}@kam.mff.cuni.cz}\and
Institute of Science and Technology Austria, Klosterneuburg, Austria\\
\email{morteza.saghafian@ist.ac.at}\and
Department of Theoretical Computer Science, Faculty of Information Technology, Czech Technical University in Prague, Czech Republic\\
\email{maria.saumell@fit.cvut.cz}\and
Computer Science Institute, Faculty of Mathematics and Physics, Charles University, Czech Republic\\
\email{josef.tkadlec@iuuk.mff.cuni.cz}}

\authorrunning{Anti\'c, Saghafian, Saumell, Schr\"oder, Tkadlec, Valtr}

\graphicspath{{Figures/}}
\DeclareMathOperator{\cross}{cr-MST}
\DeclareMathOperator{\crossAB}{cr}

\DeclareMathOperator{\CH}{CH}

\begin{document}

\maketitle

% A: Partitioning point sets to maximize the number of crossings between minimum spanning trees. Votes: Maria, Morteza, Pavel, Pepa
% B: Partitioning point sets into parts with many crossings between minimum spanning trees.
% C: How many times can two minimum spanning trees cross? Votes: Maria, Todor, Morteza, Pepa, Pavel
% --> winner is C!

% C2: How many times can two minium spanning trees cross after partitioning a planar point set?
% D: On crossings..

% A: bicolored MST crossing number vote: Todor, Maria-Morteza, Pepa
% --> winner is A!
% B: chromatic MST crossing number vote: Todor-Morteza
% C: MST chromatic crossing number vote:
% D: MST mixedness vote: Pepa-Morteza

\begin{abstract}
    Let $P$ be a generic set of $n$ points in the plane, and let $P=R\cup B$ be a coloring of $P$ in two colors. We are interested in the number of crossings between the minimum spanning trees (MSTs) of $R$ and $B$, denoted by $\crossAB(R,B)$. We define the \emph{bicolored MST crossing number} of $P$, denoted by $\cross(P)$, as $\cross(P) = \max_{P= R\cup B}(\crossAB(R,B))$. We prove a linear upper bound for $\cross(P)$ when $P$ is generic. If $P$ is dense or in convex position, we provide linear lower bounds. 
    %Lastly, if $P$ is chosen uniformly at random from the unit square, we prove that the expected number of crossings is linear.
    Lastly, if $P$ is chosen uniformly at random from the unit square and is colored uniformly at random, we prove that the expected value of $\crossAB(R,B)$ is linear.
\end{abstract}

%[We say that a point set is \textit{?tree-minimal} if each subset has a unique minimum spanning tree.]
\section{Introduction}
When two sets of points are placed in the plane, how can we tell whether they are well separated or thoroughly mixed? This intuitive question arises in a variety of contexts, from clustering and spatial statistics to ecology and, more recently, spatial biology, where understanding how different cell types are spatially arranged within tissue is of central importance, see \cite{Binnewies18,Hei12,PFRT22}.

In response to this growing interest, various measures of mixedness have been introduced. Some are classical, while others have been developed more recently, motivated by applications in spatial biology. These measures aim to capture the extent to which two or more point sets are spatially intermingled. One natural approach is to compare the connectivity structures that arise when each subset is considered independently, for instance, by constructing a spanning tree for each set and examining how those trees interact. One notable robust example that has recently attracted attention is the MST-ratio, defined as the sum of the lengths of the minimum spanning trees of the individual point sets divided by the length of the minimum spanning tree of their union. Intuitively, a higher MST-ratio indicates a greater degree of spatial mixing between the two point sets, see \cite{JMS24,CDES24,DumitrescuPach}.
%TODO: consider citing "On the coarseness of bicolored point sets"

Among such approaches, a particularly simple measure, suggested in \cite{JMS24}, considers the number of edge crossings between the minimum spanning trees (MSTs) of disjoint point sets. Intuitively, when the sets are well separated, their MSTs tend to avoid each other; conversely, the more the MSTs cross, the more spatially mixed the point sets are. This makes MST crossings a natural geometric proxy for mixedness and motivates us to study this measure more concretely. The crossings of minimum spanning trees have also been studied in~\cite{KanoMU05}, where it is proved that two MSTs cross at most $8n$ times where $n$ is the total number of points.\footnote{However, the proof of one of the key lemmas is missing from the paper. See Section 1.1 for details.} For a broader review on colored point sets, see also~\cite{kano2021discrete}.

We say that a planar point set is in \emph{general position} if no three of its points lie on a line. We say that a planar point set is in \emph{generic position with respect to minimum spanning trees}, or simply \emph{generic}, if it is in general position and any subset of it has a unique minimum spanning tree. For instance, any set in general position with pairwise distinct distances is generic. 

Let $R,B$ be two disjoint point sets in the plane such that $R \cup B$ is generic. We denote by $\crossAB(R,B)$ the number of crossings between the minimum spanning trees of $R$ and $B$. For a set $P$ of points in the plane in generic position, we denote by $\cross(P)$ the maximum value of $\crossAB(R,B)$, taken over all colorings $P=R\cup B$, and we call this parameter the \emph{bicolored MST crossing number} of $P$.

\subsection{Our Results} 

Jabal Ameli, Motiei and Saghafian proposed the study of the bicolored MST crossing number (although under a different name) in  [\cite{JMS24},Question C]. We establish the following upper and lower bounds on this parameter.

%Extending their line of inquiry, we establish the following upper and lower bounds on the bicolored MST crossing number in various settings.

\begin{restatable}{theorem}{genupperbound}
\label{thm-genupperbound}
There exists a constant $c > 0$ such that, for every generic set $P$ of $n$ points in the plane, $\cross(P) < cn$.
\end{restatable}

We remark that Theorem~\ref{thm-genupperbound} appears in~\cite{KanoMU05},\footnote{Even though it is not stated explicitly, the authors of \cite{KanoMU05} also assume that point sets have a unique minimum spanning tree.} but the authors did not include a proof of a crucial lemma [\cite{KanoMU05}, Lemma 3], and we also did not succeed in clarifying the proof with them. We include our proof of the lemma with a different constant (Lemma~\ref{lem-constantlymanycrossings}) for the sake of completeness.

\begin{restatable}{theorem}{genlowertwocols}
\label{thm-genlowertwocols}
If $P$ is a generic set of $n > 5$ points in the plane, $\cross(P)\ge 1$.
%Let $P$ be a generic set of $n > 5$ points in the plane.
%There is a coloring of $P = R \cup B$ such that $\crossAB(R, B) \ge 1$\maria{Isn't it better t say directly: $\cross(P)\ge 1$?}.
\end{restatable}

The constant $5$ in Theorem~\ref{thm-genlowertwocols} is tight, since there exists a set of $5$ points for which no coloring gives a crossing; see Figure~\ref{fig-5points}.
%Theorem~\ref{thm-genlowertwocols} implies that $\cross(P)\ge 1$ for any set of $n>5$ points. The constant $5$ is tight, since there exists a set of $5$ points for which no coloring gives a crossing; see Figure~\ref{fig-5points}.

\begin{figure}[h]
    \centering
    \includegraphics[page = 1]{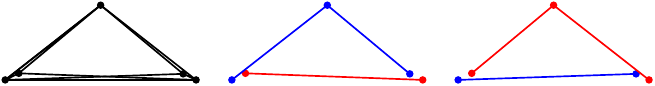}
     \caption{A set of $5$ points with bicolored MST crossing number equal to $0$. The middle and right pictures show the two MSTs for two different colorings. Interestingly, removing the top point increases the bicolored MST crossing number to $1$.}
    \label{fig-5points}
\end{figure}

% Note that $\cross(P)$  removing points from a point set $P$,  to a subset $Q\subseteq P$ might actually increase the MST chromatic crossing number.
% Note that while $\cross(P_5)=0$, removing the point $A$ we obtain a set $P_4=P_5\setminus A$ with $\cross(P_4)=1$.
Although improving the lower bound turned out to be out of reach, we can strengthen it, provided that we are allowed to discard some points of $P$ first. Note that, when a point is removed, the bicolored MST crossing number might increase -- this occurs, for instance, when removing the top point from the point set depicted in Figure~\ref{fig-5points}.
By the Erd\H os-Szekeres theorem~\cite{Erds2009}, any set of $n$ points in general position contains a set $Q$ in convex position of size $\Omega(\log n)$, which satisfies $\cross(Q)=\Omega(\log n)$ (see Theorem~\ref{thm-convex} below). We can prove an asymptotically stronger lower bound.
%Since bridging the difference between our upper and lower bounds has proven to be out of reach, we consider the following relaxation of the problem. We consider a coloring\pepa{Maybe instead phrase this as ``when allowed to discard some points'', that would nicely follow on Figure~\ref{fig-5points}. I will do that} of a point set $P$ into three colors, $P=R\cup B \cup G$ and we count the minimal number of crossing between any two MSTs of $R$ and $B$, we denote this quantity $\crossRBG(R,B,G)$ and the maximum value of $\crossRBG(R,B,G)$ over all such colorings of $P$ by $\crossthree(P)$. Note that the upper bound from Theorem~\ref{thm-genupperbound} still holds. \pepa{Maybe mention here that there is a simple lower bound $\Omega(\log n)$ by Erd\H os-Szekeres?} But we can prove a substantially\pepa{maybe ``asymptotically''?} stronger lower bound in this setting. 

\begin{restatable}{theorem}{genlowerthreecols}
\label{thm-genlowerthreecols}
Let $P$ be a generic set of $n$ points in the plane. There exists a subset $Q\subseteq P$ such that
$\cross(Q) \in \Omega(\log^2 n)$.
%Then $\crossthree(P) \in \Omega(\log^2 n)$.
\end{restatable}

%\maria{Shouldn't we have first all results on $\cross(P)$ and then all results on $\crossthree(P)$?}
%\textcolor{blue}{Todor: how I view it, we first deal with the problem in general (but we can't do it so well) so we only consider this variation in general position. For other cases we have linear lower bounds.}

%We also show linear lower bounds for $\cross(P)$ if: (i) $P$ is in convex position; (ii) $P$ is a perturbation of a $2\times \frac{n}{2}$ grid; (iii) $P$ is a set of points in a square chosen uniformly at random; (iv) the ratio between the maximum and minimum distances between the points in $P$ is bounded (we call such sets \emph{dense}).  

In the original setting where we are not allowed to discard points, we also improve our lower bound for $\cross(P)$ to linear  in four particular cases described in the theorems below. First, we focus on points sets in convex position.

\begin{restatable}{theorem}{convex}
\label{thm-convex}
Let $P$ be a generic set of $n$ points in the plane in convex position. Then $\cross(P) \ge \left\lfloor \frac{n}{2} \right\rfloor - 1$.
\end{restatable}

Next, we consider perturbations of the $2\times \frac{n}{2}$ grid. For this family of point sets, we establish a tight lower bound for $\cross(P)$.
%The next theorem tells us that we cannot hope to match the above bound for generic convex point sets, by perturbing a $2\times \frac{n}{2}$ grid.

\begin{restatable}{theorem}{grids}\label{thm-grids}
Let $P$ be a perturbation of a $2\times \frac{n}{2}$ grid (where $n\in \mathbb{N}$ is even) that is generic and in convex position. Then, $\cross(P)\geq \lfloor\frac{5}{8}(n-2)\rfloor$.
    Further, there is a perturbation $P_0$ of the grid that is generic and in convex position such that $\cross(P_0) \le \lfloor\frac{5}{8}(n-2)\rfloor$.
%    Let $P$ be a perturbation of a $2\times \frac{n}{2}$ grid (where $n\in \mathbb{N}$ is even) that is generic and in convex position. Then, there is a coloring of $P=R\cup B$ such that $\crossAB(R,B)\ge \lfloor\frac{5}{8}(n-2)\rfloor$.\maria{Isn't it better to say: Then, $\cross(P)\geq \lfloor\frac{5}{8}(n-2)\rfloor$?}
 %   Further, there is a perturbation $P_0$ of the grid \maria{shall we add generic and in convex position?} such that $\cross(P_0) \le \lfloor\frac{5}{8}(n-2)\rfloor$.
\end{restatable}

We note that the condition of being generic is crucial in Theorem~\ref{thm-grids}. In fact, if we allow non-generic point sets and, in case of multiplicity of MSTs, we choose the two MSTs that minimize the number of crossings, then the bound in Theorem~\ref{thm-convex} is tight; see Lemma~\ref{lem-nonUniqueMST}.

%This is tight:\maria{not true?}

%\begin{figure}[htb]
 %   \centering
  %  \includegraphics[width=0.8\linewidth]{Figures/halfNCrossingsConvex.pdf}
   % \caption{Coloring witnessing tightness of Theorem~\ref{thm-convex}.\maria{One of the points has wrong color?}}
    %\label{fig-convex}
%\end{figure}

 Next, we look at random points uniformly distributed in $[0,1]^2$. 
 
\begin{restatable}{theorem}{random}
\label{thm-random}
There exists a constant $c>0$ satisfying the following: Let $P$ be a set of $n$ random points uniformly distributed in $[0,1]^2$. If we color $P$ in two colors $P=R \cup B$ uniformly at random, then $\mathbb{E}[\crossAB(R,B)] \ge cn$.
\end{restatable}

Finally, we focus on the so-called \emph{dense} point sets, that is, (generic) sets with bounded ratio between the maximum and minimum distances defined by the points.
Various problems for dense point sets have been studied since the late 1980's \cite{Alon1989,Bukh2025,DumiPachGD,DumitrescuPach,DumiToth2022,Edelsbrunner1997,Kovcs2019,Valtr1992,valtr1994,Valtr1996}.

Formally, let $P$ be a set of $n$ points in generic position in the plane such that $\min_{a,b\in P}(d(a,b))=1$ (this can always be achieved by rescaling the set). The diameter of $P$ is known to be at least $c\sqrt{n}$, where $c>0$ is a constant. If $n$ is sufficiently large, $c\ge \frac{\sqrt{2}\sqrt[4]{3}}{\sqrt{\pi}}$ \cite{Valtr1992}. We say that $P$ is \textit{$\alpha$-dense} if the diameter of $P$ is less than $\alpha \sqrt{n}$, for some constant $\alpha>0$. 
%With this definition, we can prove the following theorem.

\begin{restatable}{theorem}{dense}
\label{thm-dense}
Let $P$ be a generic $\alpha$-dense set of $n$ points for $\alpha \ge \frac{\sqrt{2}\sqrt[4]{3}}{\sqrt{\pi}}$. If $n$ is sufficiently large, then $\cross(P)\ge cn$ for some constant $c>0$ that depends only on $\alpha$.
\end{restatable}

\subsection{Notation}

In this paper, we consider only finite sets of points in the plane. For brevity, we sometimes refer to them as point sets. Unless we explicitly state otherwise, we use $\log$ to refer to the logarithm with base $2$. For a point set $X$, we use $\CH(X)$ to denote the convex hull of $X$ and $\partial\CH(X)$ to denote the boundary of $\CH(X)$. If $X$ is generic, we use $T_X$ to denote the MST of $X$.

\section{The general case}\label{sec:general}
We begin this section with the proof of Theorem~\ref{thm-genupperbound}. Before this, we need the following two technical lemmas. 

\begin{restatable}{lemma}{crossingssmallangle}
\label{lem-crossingssmallangle} Let $a,b,c,d \in \mathbb{R}^2$ be four points such that the line segments $ab$ and $cd$ cross at a point $x$. Suppose that $|ab|,|cd|\ge 2$, $|ax|\le |xb|$, $|cx| \le |xd|$, and $\angle axc\le 1^\circ$. Then $\max\{|ab|,|cd|\} \ge \max\{|ac|,|bd|\}+0.5$.

\end{restatable}

\begin{proof}
    We argue that the inequality holds for both $|ac|$ and $|bd|$. See Figure~\ref{fig-crossingssmallangle}. For $|ac|$: Since $\angle axc\le 1^\circ <60^\circ$, in the triangle $\triangle acx$ the angle by $x$ is not the largest, so $ac$ is not the longest side. Without loss of generality, suppose $|ac|<|ax|\le \frac12|ab|$. Then $|ab|-|ac|\ge \frac12|ab|\ge 1>0.5$ as needed.

\begin{figure}[h!]
    \centering
    \includegraphics[page=1]{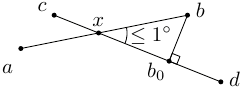}
    \caption{Proof of Lemma~\ref{lem-crossingssmallangle}.}
    \label{fig-crossingssmallangle}
\end{figure}

    For $|bd|$: Similarly, without loss of generality suppose $xd$ is the longest side in the triangle $\triangle xbd$. Let $b_0$ be the projection of $b$ onto the segment $xd$. By triangle inequality in $\triangle bb_0d$ we have $|bd|\le |bb_0|+|b_0d|$, so using the right triangle $\triangle bb_0x$ we get $|xd|-|bd|\ge |xb_0|-|bb_0|\ge |bx|\cdot(\cos 1^\circ - \sin 1^\circ)>1\cdot (0.99-0.02)>0.5 $. \qed
\end{proof}

\begin{lemma} \label{lem-constantlymanycrossings}
    Let $P,S$ be two disjoint point sets. Then there is a constant $c>0$, independent from $P$ and $S$, such that any edge $e$ of $T_P$ crosses at most $c$ edges of $T_S$ whose length is at least the length of $e$. 
\end{lemma}
%TODO: go through lemma's statements once again and check if the generic condition should be added.
\begin{proof}
    Let $e$ be an edge of $T_P$ and assume that $|e|=2$ and that $e$ is horizontal. Let $C>2$ be a constant. Assume that $e$ is crossed by at least $720C +2$ edges of $T_S$ whose length is at least the length of $e$. Then, at least $360 C+1$ of those edges have half or more of their length in the same halfplane determined by the line through $e$.
    %\maria{in the rest of the proof are we making any assumption on which halfplane it is? if yes, let's state it; if not, would it help?} 
    Associate to each such edge $e'$ a tuple $(\chi,\alpha)$, where $\chi \in(0,2)$ is the distance from the left endpoint of $e$ to the crossing point between $e$ and $e'$ and $\alpha \in (0,\pi)$ is the size of the angle defined by the portion of the plane lying above $e$ and to the right of $e'$. Now, by the pigeonhole principle we can find two such edges $e_1=uv,e_2=xy$ with tuples $(\chi_1,\alpha_1)$ and $(\chi_2,\alpha_2)$ such that $|\chi_1-\chi_2|\leq \frac{1}{C}$ and $|\alpha_1-\alpha_2|\leq 1^{\circ}$. In particular, the lines through $e_1$ and $e_2$ cross at an angle less than $1^{\circ}$. Let us assume that $u$ and $x$ are on the same side of $e$ and that this is the side which contains the shorter portions of $e_1,e_2$. Now, translate the edge $e_2=xy$ by distance at most $\frac{1}{C}<0.5$, so that it crosses $e$ at the same point as $e_1$. Call this translate $x'y'$.  By Lemma~\ref{lem-crossingssmallangle} we see that
    \[|ux| \le |ux'| + |xx'| \le \max\{|uv|,|x'y'|\}=\max\{|uv|,|xy|\}.\]
    Similarly, we see that 
    \[|vy| \le |vy'| + |y'y| \le \max\{|uv|,|x'y'|\}.\]
    Therefore, one of the edges $uv,xy$ is the longest edge of the cycle $uvyx$ in the complete geometric graph on $S$ and cannot be in $T_S$, which is a contradiction. \qed
\end{proof}
%TODO: Perhaps a small comment about what happens when there are repetitions in $u,v,x,y$

We remark that the constant that we obtain from Lemma~\ref{lem-constantlymanycrossings} is at least $1442$, which is significantly higher than the constant~$8$ claimed by the authors of~\cite{KanoMU05}. Unfortunately, we have not been able to reproduce their result, but we note that a more careful analysis in Lemma~\ref{lem-crossingssmallangle}, can lead to a better constant (although still far from $8$). 

We now finalize the proof of Theorem~\ref{thm-genupperbound}. 
\genupperbound*
\begin{proof} 
    Consider a coloring $P=R\cup B$. For each edge $e$ in $T_R$, let $\phi_R(e)$ be the number of crossings between $e$ and an edge $e'$ of $T_B$ such that $|e'|\ge |e|$. We define $\phi_B$ similarly. Clearly, $\crossAB(R,B) \le \sum_{e\in T_R}\phi_R(e) + \sum_{e\in T_B}\phi_B(e)$. By Lemma~\ref{lem-constantlymanycrossings}, $\phi_R$ and $\phi_B$ are upper bounded by a constant $c$, so the result follows.\qed
\end{proof}

Next, we prove Theorem~\ref{thm-genlowertwocols}. 
\newpage
\genlowertwocols*
\begin{proof}
     We split into two cases, based on the number of vertices on the convex hull of $P$. 
     \begin{description}
         \item[Case 1] The convex hull of $P$ has at least $4$ vertices. \\
    We color two non adjacent vertices of the convex hull of $P$ red and color all of the other points of $P$ blue. Then $T_R$ consists of only one edge that splits $B$  into two parts, so at least one edge of $T_B$ crosses it. See Figure~\ref{fig-generallowertwocols} (left).   
    \item[Case 2] The convex hull of $P$ has three vertices. \\
    Let $a,b,c$ be the three vertices of the convex hull of $P$, and assume that $ab$ is at least as long as the other edges of the convex hull. Additionally, assume that $ab$ is horizontal and that $a$ is to the left of $b$. 
    \begin{description}
        \item[Case 2.1] There is a point $x\in P$, $x\neq c$, such that $\triangle axb$ contains a point of $P$ in the interior. \\
    We color $a,x,b$ red and all other points of $P$ blue. The tree $T_R$ has two edges: $ax$ and $bx$. Let $Y$ be the set of points of $P$ in the interior of $\triangle axb$. At least one edge of $T_B$ connects a (blue) point of $Y$ to a blue point that is not in $Y$, and this edge crosses either $ax$ or $bx$. See Figure~\ref{fig-generallowertwocols} (middle). 
    \item[Case 2.2] $\triangle axb \cap (P\setminus\{a,x,b\}) = \emptyset,  \forall x\in (P \setminus \{a,b,c\})$. \\
      Note that the assumptions of this case guarantee that no two points in $P\setminus \{a,b,c\}$ have the same $x$-coordinate, as otherwise we would be in Case 2.1. With this in mind, sort the points in $(P \setminus \{a,b,c\})$ by their $x$-coordinate.  Let $v$ be either the middle point or (if $n$ is odd) one of the two middle points in this ordering. Since $n>5$, there are points of $(P \setminus \{a,b,c\})$ both to the left and right of $v$. Assume that the $x$-coordinate of $v$ is smaller than or equal to the $x$-coordinate of $c$. This implies that $ac$ is the longest edge of $\triangle(a,v,c)$. Now, we color $a,v,c$ red and all the other points blue. By the same argument as in Case 2.1, there is a crossing between $T_R$ and $T_B$. See Figure~\ref{fig-generallowertwocols} (right).
      \end{description}
      \end{description}
      \qed
\end{proof}

\begin{figure}
    \centering
    \includegraphics[width=0.9\linewidth,page=2]{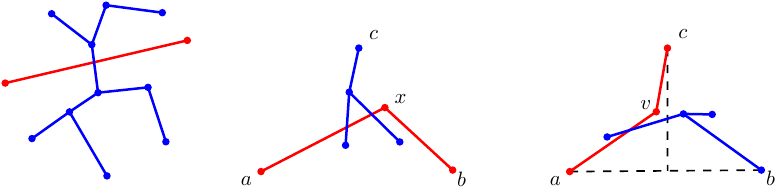}
    \caption{Colorings of the proof of Theorem~\ref{thm-genlowertwocols}.}
    \label{fig-generallowertwocols}
\end{figure}

Note that $\cross(P)$ is not monotone in $n = |P|$. For example, removing the topmost point in Figure~\ref{fig-5points} yields a point set for which the bicolored MST crossing number is $1$, but adding this point back brings the bicolored MST crossing number to $0$. This might be an indicator that finding better bounds for $\cross(P)$ is a difficult problem and further structural analysis might be needed to improve Theorem~\ref{thm-genlowertwocols}.

We conclude this section by considering a relaxed version of the problem where we are allowed to discard some points before coloring the rest in two colors.
%use three colors but only count the crossings between the two of them (in a sense, we are allowed to ``throw away'' some of the points).
Before proving Theorem~\ref{thm-genlowerthreecols}, we need a technical lemma.

\begin{restatable}{lemma}{islandsdonotinteract}
\label{lem-islandsdonotinteract}
Let $P = P_1 \cup P_2$ be a set of points in the plane such that $P_1$ is contained in a unit disk centered at point $o$, and each point $a \in P_2$ satisfies $|oa| > 3$ and is contained in a fixed wedge with apex at $o$ and angle $3.6^\circ$. Then, $T_P$ contains only one edge  connecting $P_1$ to $P_2$.
\end{restatable}

\begin{proof}
    For contradiction, suppose that the MST of $P$ contains two such edges and denote them by $ab$, $cd$ with $a,c\in P_1$ as in Figure~\ref{fig-islandsdonotinteract}. To reach the contradiction, it suffices to show that both $ac$ and $bd$ are shorter than $\max\{|ab|,|cd|\}$.
    
    \begin{figure}[h]
    \centering
    \includegraphics{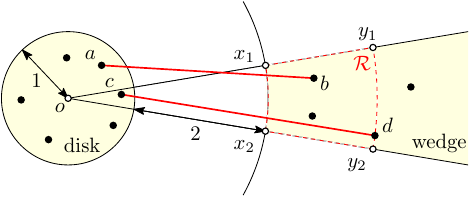}
    \caption{Proof of Lemma~\ref{lem-islandsdonotinteract}. The wedge is drawn with a larger angle for readability.}
    \label{fig-islandsdonotinteract}
    \end{figure}

    Since $|ob|,|od|> 3$ and $|oa|,|oc|\le 1$, we have $|ab|,|cd|>2$ and $|ac|\le 2$. In particular, $|ac|< \max\{|ab|,|cd|\}$ .

    Regarding $|bd|$, without loss of generality assume that $|od|\ge |ob|$ and denote $|od|=3+\varphi$, where $\varphi>0$. Then $|cd|\ge |od|-|oc|\ge2+\varphi$. Moreover, $b$ lies in a quadrangular region $\mathcal R$ delimited by the sides of the wedge, and by circles centered at origin with radii 3 and $3+\varphi$. Denote the respective intersection points by $x_1$, $x_2$, $y_1$, $y_2$ as shown in Figure~\ref{fig-islandsdonotinteract}. Note that the diameter of $\mathcal R$ is equal to either $|y_1y_2|$ or to $|x_1y_2|=|x_2y_1|$, thus $|bd|\le \max\{|y_1y_2|,|x_1y_2|\}$.
    
    From the isosceles triangle $oy_1y_2$ we have $|y_1y_2|=2\cdot (3+\varphi)\cdot \sin(1.8^\circ)<0.2\cdot (3+\varphi)$  and likewise $|x_1x_2|<0.2\cdot 3$. Since $0.2\cdot (3+\varepsilon)<2+\varepsilon$ for all $\varepsilon\ge 0$, we thus get $|y_1y_2|<|cd|$. Moreover, from triangle inequality in $\triangle x_1x_2y_2$ we have $|x_1y_2|<|x_1x_2|+|x_2y_2|<0.2\cdot 3+\varphi<2+\varphi\le|cd|$. In total, $|bd|<\max\{|y_1y_2|,|x_1y_2|\}<|cd| \leq \max\{|ab|,|cd|\}$, as required. \qed
\end{proof}

We now prove Theorem~\ref{thm-genlowerthreecols} using Lemma~\ref{lem-islandsdonotinteract} and Theorem~\ref{thm-convex}, which will be proved in a later section. 

\genlowerthreecols*
\begin{proof}
    Assume that $100^{k+1} \ge  n\ge 100^k$ for some natural number~$k$. We describe a coloring procedure in stages that gives us our desired number of crossings. The points colored green are to be discarded. Let $S$ be the set of currently colored points. We begin the process with $S=\emptyset$, $m_1=|P\setminus S|= n$ and $q_1=\sqrt{m_1}$. 
    
    In the $i^{\text{th}}$ step of the procedure (illustrated in Figure~\ref{fig-threecols}), let $m_i = |P\setminus S|,q_i=\sqrt{m_i} $ and let $C_i$ be a smallest circle containing $q_i$ points of $P\setminus S$. We call $C_i$ the $i$-th \emph{island} and denote its center $x_i$ and its radius $r_i$. By the Erd\H{o}s Szekeres theorem \cite{Erds2009}, there is a set of $\Omega(\log q_i) = \Omega(\log m_i)$ points in convex position inside $C_i$. We color these points in red and blue in such a way that colors alternate along the boundary of the convex hull; in this way, we obtain $c \log m_i$ crossings between $T_R$ and $T_B$ inside $C_i$ (for details, see Theorem~\ref{thm-convex}). We color all the other points inside $C_i$ green. Let $C'_i$ be the disk concentric with $C_i$ and with radius $3r_i$. We color all of the remaining uncolored points inside $C'_i$ green. Note that $C'_i \setminus C_i$ can be covered by $30$ copies of $C_i$~\cite{FriedmanURL}; thus, $C'_i \setminus C_i$ contains at most $30q_i$ points.  Finally, we consider a subdivision of the plane into $100$ equal wedges centered at $x_i$. Let $W_i$ be the wedge with the highest number of uncolored points; note that $|W_i| \ge \frac{|P \setminus S|}{100}$. We color the points inside all other wedges green and proceed to the next step with $m_{i+1} = |(P\setminus S ) \cap W_i| \ge \frac{m_i - 31q_i}{100}$ and $q_{i+1} = \sqrt{m_{i+1}}$; see Figure~\ref{fig-threecols}. We repeat this procedure until all of the points of $P$ are colored. 
    
        \begin{figure}[h!]
        \centering
        \includegraphics[page=3]{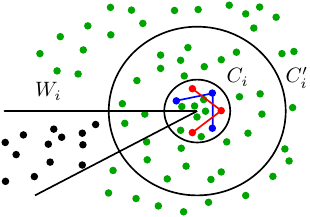}
        \caption{Situation after the $i^{th}$ step of the procedure in the proof of Theorem~\ref{thm-genlowerthreecols}. Wedge $W_i$ is drawn with a larger angle for readability.}
        \label{fig-threecols}
    \end{figure}

   Next, denote by $T_R[C_i]$ the MST of only those red points that belong to the island $C_i$. We next show that each $T_R[C_i]$ is contained in $T_R$. An analogous statement holds for the blue set.
    
    Suppose, for the sake of contradiction, that an edge $ab\in T_R[C_i]$ is not an edge of $T_R$. Take the path $p$ from $a$ to $b$ in $T_R$. If $p$ never leaves $C_i$, then $p$ together with $ab$ forms a cycle within $C_i$ whose longest edge is $ab$; this contradicts the fact that $ab\in T_R[C_i]$. Otherwise, let $j$ be the smallest index such that $C_j$ is visited by $p$ (note that $j$ might be equal to $i$). Then the path $p$ includes two edges connecting $C_j$ to one of the subsequent islands, a contradiction with Lemma~\ref{lem-islandsdonotinteract}.  %\maria{A bit more is needed here, because if $j>i$ this argument doesn't work (in this case, the problem is at $C_i$ instead of $C_j$)}\todor{I don't think $j>i$ is possible because $j$ is defined to be mininmal and $p$ connects two points in island $i$ so by default $j\le i$, right? If I'm correct this comment can be deleted.}\maria{I don't understand... Why isn't it possible that $p$ begins in $C_i$, then it moves to some $C_j$ with $j>i$, and then it returns to $C_i$?}\todor{It is possible, but then j is not minimal index such that $C_j$ is visited by $p$, since $p$ visits $C_i$ and $i<j$, right?}\maria{Right. So what I would do is, when we define $j$, add this in brackets at the end: (notice that $j$ might be equal to $i$)} 
    
    %By Lemma~\ref{lem-islandsdonotinteract}, we know that there is at most one edge between the points inside $C_i$ and $C_j$ for $i\neq j$, hence, the edges of the red and blue MSTs inside $C_i$ are also present in the red and blue MSTs for the whole set.
    
    It only remains to count the crossings between $T_R$ and $T_B$. In each step, we get $\Theta (\log m_i)$ crossings. Additionally, since $q_i \in o(m_i)$, we perform logarithmically many steps before we run out of points and terminate the procedure. \qed
\end{proof}

\section{Special cases}

%%%%%%%%%%%%%%%%%%%%%%%%%%%%%%%
\subsection{Convex Point Sets}
\convex*
\begin{proof} 
We color $P$ using an alternating coloring along $\partial \CH(P)$ (see Figure~\ref{fig-convex}). Then $||R|-|B|| \in \{0,1\}$. Assume that $|R|\ge |B|$. Then $\redmst$, together with $\partial \CH(P)$, splits $CH(P)$ into cells such that no two blue points are in the same cell. Thus, any edge of $\bluemst$ crosses at least one edge between two cells, and this edge is necessarily an edge of $\redmst$. Since $\bluemst$ has $\lfloor\frac{n}{2}\rfloor-1$ edges, we conclude that $\crossAB(R,B)\geq \lfloor\frac{n}{2}\rfloor-1$. \qed
\end{proof}

\begin{figure}[h]
    \centering
    \includegraphics[page=4]{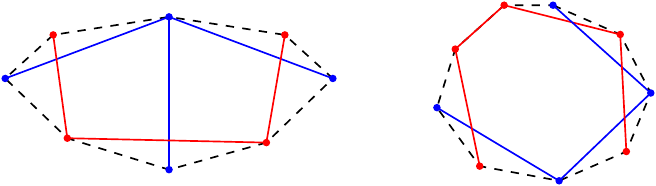}
    \caption{Coloring of the proof of Theorem~\ref{thm-convex}.}
    %TODO: Add comment that in the left figure the bound is tight?
    \label{fig-convex}
\end{figure}

We can also prove a version of Theorem~\ref{thm-convex} when the coloring is randomized. 

\begin{lemma}
    Let $P$ be a set of $n$ points in convex position. If we color $P$ with two colors $P=R\cup B$ uniformly at random, then  $\mathbb{E}[\crossAB(R,B)]\ge n/4-1$. 
\end{lemma}
\begin{proof}
    Denote the points of $P$ as they appear along $\partial \CH(P)$ by $v_1,v_2,\dots, v_n$. Let $X_i$ be the event that $v_{i+1}$ is colored differently from $v_i$, with indices counted modulo $n$. Clearly, $\mathbb{P}[X_i]=\frac{1}{2}$. Thus, by the linearity of expectation, $\mathbb{E}(\sum I(X_i))=\frac{n}{2}$, where $I$ is the indicator function.
    If $\sum I(X_i)=2k>0$, there are exactly $k$ maximal blue intervals along $\partial \CH(P)$. To connect all these intervals into one tree ($\bluemst$), at least $k-1$ edges of the tree  connect vertices from different intervals.
    Each such edge crosses at least one edge of $\redmst$. Hence, there are at least $k-1$ crossings between the two MSTs. It follows that the expected number of crossings is at least $\mathbb{E}(k-1)=\frac12\mathbb{E}(\sum I(X_i))-1=\frac n4-1.$ \qed
\end{proof}

Note that, when $P$ is a perturbation of a regular $n$-gon, the alternating coloring from Theorem~\ref{thm-convex} yields MSTs that are paths, and we get $\cross(P)= n-O(1)$. This is roughly twice as many crossings compared to what is guaranteed by Theorem~\ref{thm-convex} for any convex point set.
Next, we show that the same situation occurs at the other extreme, that is, for point sets that are very ``flat'' rather than very ``round''.
%\pepa{I will comment on round vs very flat both $\sim n$.} We now present a specific class of point sets in convex position for which Theorem~\ref{thm-convex} is very far from being optimal.
Here we call a point set in convex position \emph{flat}, if the maximum difference in $y$-coordinates between any two points is smaller than the minimum difference in $x$-coordinates.

\begin{restatable}{lemma}{convexflat}
\label{lem-convex-flat}
If $P$ is a flat convex point set with $n$ points, $\cross(P) \ge n - 3$.
\end{restatable}

\begin{proof}
    %\pepa{Sketch, I will polish it:} %[Maybe define ``flat convex'' instead of just ``flat'']
    Label the points $v_1,v_2,\dots,v_n$ in the order from the left to the right. Since $P$ is convex, each point $v_i$ with $2\le i\le n-1$ belongs either to the upper arc or to the lower arc of the convex hull of $P$.
    
    We color the points from left to right as follows: First, color $v_1$ in red and $v_2$ in blue. Then, for $3\le i \le n-1$, use the following rule: If $v_i$ and $v_{i-1}$ lie on different arcs, color $v_i$ the \emph{same} as $v_{i-1}$. If they lie on the \emph{same} arc, color $v_i$ with the \emph{other} color. Finally, color $v_n$ with the color different from that of $v_{n-1}$. See Figure~\ref{fig-convex-flat}.
    
        \begin{figure}[h!]
            \centering
            \includegraphics{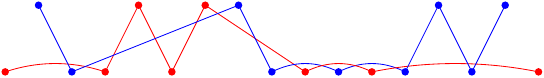}
            \caption{Proof of Lemma~\ref{lem-convex-flat}.}
            \label{fig-convex-flat}
        \end{figure}

    Now we argue that $\redmst$ and $\bluemst$ cross  $n-3$ times. Since the point set is flat, both MSTs are paths that connect the points in the order from the left to the right.
    %We will charge each crossing of the MSTs to a different point from among $A_3,A_4,\dots,A_{n-1}$. Suppose $A_iA_j$ and $A_kA_l$ are crossing MST edges with $i<j$ and $k<l$. Assign their crossing to $A_{\min \{k,l\}}$. Then 
    We count the crossings along $\bluemst$. Note that $\bluemst$ consists of segments $v_iv_k$ that either ``stay'' within one arc, or that ``jump'' to the other arc. Either way, consider all crossings on one such segment, charge the rightmost one to the point $v_k$, and charge each other one to the right endpoint of the corresponding red edge. In this way, each red point $v_j$ with $i<j<k$ gets charged exactly once, and the blue point $v_k$ gets charged too. Repeating this for all blue segments, we charge once each point of $P$, except for the two leftmost points $v_1$, $v_2$, and the (unique) red point from among the two rightmost points $\{v_{n-1},v_n\}$. Therefore, there are $n-3$ crossings. \qed
\end{proof}

%%%%%%%%%%%%%%%%%%%%%%%%%%%%%%%%%%%%%%%%%%%%%%%%%%%%%%%%%%%%
%%%%%%%%%%%%%%%%%%%%%%%%%%%%%%%%%%%%%%%%%%%%%%%%%%%%%%%%%%%%
\subsection{Perturbations of the $2 \times \frac{n}{2}$ grid}
%I think the solutions is great :)
%\maria{Might be nice to add some sentence explaining what motivated us to study these grids} \todor{Do you have some suggestion what to write?} \maria{What about rephrasing the first paragraph as follows:}

In this section, we use a perturbation of the $2 \times \frac{n}{2}$ grid to prove that Theorem~\ref{thm-convex} is best possible in the non-generic setting (Lemma~\ref{lem-nonUniqueMST}). Afterwards, we consider the same family of point sets in the generic setting, and we show that in this case it cannot be used to prove that Theorem~\ref{thm-convex} is best possible. In fact, for generic perturbations of the $2\times \frac{n}{2}$ grid, the minimum value of the bicolored MST crossing number is exactly $\lfloor\frac{5}{8}(n-2)\rfloor$  (Theorem~\ref{thm-grids}). 

%In this section we present two results. 
%First, we prove that Theorem~\ref{thm-convex} is best-possible in the non-generic setting (Lemma~\ref{lem-nonUniqueMST}). 
%%First, we prove a matching lower bound for Theorem~\ref{thm-convex} in the non-generic setting (Lemma~\ref{lem-nonUniqueMST}). 
%Afterwards we show that, for generic perturbations of the $2\times \frac{n}{2}$ grid, the minimum value of the bicolored MST crossing number 
%is exactly $\lfloor\frac{5}{8}(n-2)\rfloor$  (Theorem~\ref{thm-grids}). 

Let us first introduce some notation that we will use throughout this section. For us, the $2\times \frac{n}{2}$ grid consists of a row of grid points $v_1=(1,0),v_2=(2,0),\ldots,v_{n/2}=(n/2,0)$ at $y=0$ and a row of grid points $w_1=(1,1),w_2=(2,1),\ldots,w_{n/2}=(n/2,1)$ at $y=1$.
We say that a $2\times \frac{n}{2}$ grid has $2$ \emph{rows} and $\frac{n}{2}$ \emph{columns}. With some abuse of notation, even after perturbing the grid, we still denote the points  by $v_1,\ldots,v_{n/2}$ and $w_1,\ldots,w_{n/2}$ and talk about rows and columns. We call an edge between two points in a perturbed $2 \times \frac{n}{2}$ grid \emph{vertical} if it connects two points in the same column, \emph{horizontal} if it connects two points in the same row, and \emph{diagonal} otherwise. 

Suppose that we are given a coloring of the points of a (perturbed) $2 \times \frac{n}{2}$ grid. A column is \emph{rainbow} if the two points of the column receive different colors, and it is \emph{monochromatic} if they receive the same color. Additionally, we call a crossing between  $T_R$ and $T_B$ \emph{vertical} if one of the edges involved in the crossing is vertical, and \emph{non-vertical} otherwise.

Before proving that Theorem~\ref{thm-convex} is tight in the non-generic setting, let us formally define the bicolored MST crossing number for non-generic sets. Recall that non-generic sets might have more than one MST. For a non-generic set $P$ and a coloring $P=R\cup B$, we define $\crossAB(R,B)$ to be the minimum number of crossings among all choices of MSTs for $R$ and $B$. Then, $\cross(P)$ is again defined as the maximum value of $\crossAB(R,B)$ taken over all colorings.

\begin{restatable}{lemma}{nonUniqueMST} \label{lem-nonUniqueMST}
Let $P$ be a convex point set obtained by perturbing the $2 \times \frac{n}{2}$ grid in such a way that each row of the grid becomes a set of equidistant points on a convex curve, and the nearest neighbor of each vertex is its column neighbor. 
Then $\cross(P) \leq \frac{n}{2} - 1$.
%Then any coloring of $P$ gives at most $\frac{n}{2}-1$ MST crossings.
\end{restatable}

\begin{proof}
    Suppose a coloring of $P$ with red and blue is given. Since the point set is non-generic, $R$ and $B$ might have several MSTs. 

We define ${\cal T}_B$ (resp., ${\cal T}_R$) as the set containing all MSTs of $B$ (resp., $R$). Observe that, for any column, if the two points have  color blue (resp., red), then the edge connecting them is contained in any  $T_B\in {\cal T}_B$ (resp.,  $T_R\in {\cal T}_R$).

Let $i_1,\dots,i_k$ be the indices of rainbow columns different from $C_1,C_{n/2}$, listed from left to right. They split the set $P$ into $k+1$ parts $P_0,\dots,P_k$, such that $P_0$ is the union of the columns $C_1,\dots,C_{i_1-1}$, each $P_j$ with $1\leq j< k$ is the union of the columns $C_{i_j},\dots,C_{i_{j+1}-1}$, and
$P_k$ is the union of the columns $C_{i_k},\dots,C_{n/2}$.
Since each column is in exactly one section, the sections $P_0,\dots,P_k$ have $n/2$ columns in total. 

The following claims are stated for color blue, but they also hold for color red.

\begin{restatable}{ourclaim}{edgescrossinghalfint}\label{cl:edges-crossing-half-integer-vertical-lines-first} Let $i \in \{1,\dots,n/2-1\}$ and  $T_B\in {\cal T}_B$. Suppose that two distinct edges $e,f$ of  $T_B$ intersect the vertical line $x = i+1/2$.
Then all four points $v_i, v_{i+1}, w_i, w_{i+1}$ are blue and the edges $e,f$ are the edges $v_i v_{i+1}$ and $w_i w_{i+1}$. 
%An analogous statement holds for $\redmst\in {\cal T}_R$.
\end{restatable}

\begin{proof}
    Let $e=ab$ and $f=a'b'$, where $a,a'$ lie to the left of the line $x=i+1/2$ and $b,b'$ lie to the right of it.
If $a=a'$ then, due to the construction our point set, one of the edges $e,f$ is the longest edge of the triangle $\triangle abb'$ with all vertices colored blue. Kruskal's algorithm then gives a contradiction.
%\pavel{Maybe we could make separate lemma about the longest edge of a cycle not lying in an MST?}\maria{I agree. It is probably used in other parts of the paper, so we should find the best place for it}
The case $b=b'$ is analogous. Suppose now that the edges $e,f$ have no common vertex and they are not exactly the edges
$v_iv_{i+1}$ and $w_iw_{i+1}$. Since the edges $v_iw_{i+1}$ and $w_iv_{i+1}$ cross each other and therefore cannot be both in the blue MST, at least one of the points $a,a'b,b'$, is different from each of the four points closest to the line $x=i+1/2$ (the points $v_i,v_{i+1},w_i,w_{i+1}$). Then one of the edges $e,f$ is the longest of the blue $4$-cycle $abb'a'$, which again gives a contradiction due to the Kruskal's algorithm. \qed
\end{proof}

Claim \ref{cl:edges-crossing-half-integer-vertical-lines-first} implies that every vertical edge is involved in at most one vertical crossing. Hence, no matter how the two MSTs are chosen, in the convex hull of any part $P_j$ with $c(j)$ columns, there are at most $c(j)-1$ vertical crossings: There is no vertical crossing involving a  possible vertical edge connecting $v_{i_j}$ and $w_{i_j}$ because the first column of the part is rainbow or, in the case of $P_0$, $C_0$.
For $P_k$, there is also no vertical crossing involving a possible edge connecting $v_{n/2}$ and $w_{n/2}$,
thus the convex hull of $P_k$ contains at most $c(k)-2$ crossings.
In the remainder of the proof, we show that it is possible to choose the two MSTs in such a way that the remaining non-vertical crossings can be assigned to parts so that at most one such crossing is assigned to each part. This implies our desired bound, since the total number of crossings will be at most $(c(0)-1+1)+(c(1)-1+1)+\dots+(c(k)-2+1)=\sum c(i)-1=n/2-1.$ 

We start with some further observations on the structure of the two MSTs. 

\begin{restatable}{ourclaim}{sectionshorizontalvertical} \label{cl:edges-in-sections-horizontal-or-vertical} Let $j\in \{0,1,\dots, k\}$ and  $T_B\in {\cal T}_B$. 
%and let $\bluemst$ be some MST of $B$. 
All edges of $\bluemst$ with both endpoints in $P_j$ are either horizontal or vertical. 
%An analogous statement holds for $\redmst$. 
\end{restatable}

\begin{proof}
    Assume, for the sake of contradiction, that there is a diagonal edge $v_aw_b$ for some $a,b$ such that $C_a,C_b \in P_j$. We can assume that $a<b$, and therefore $C_b$ is a monochromatic (blue) column. Then, $v_aw_b$ is the longest edge of the blue 3-cycle $v_aw_bv_b$, and thus is not an edge of $\bluemst$.   \qed
\end{proof}

\begin{ourclaim} \label{cl:choice-horizontal-edges} 
    Let $j\in\{0,1,\dots,k\}$ and let $a,b$ be such that $C_a,C_b \in P_j$ and $v_a,v_b,w_a,w_b$ are all blue. If there exists some $\bluemst\in {\cal T}_B$ containing the edge $v_av_b$, then the tree obtained by replacing $v_av_b$ by $w_aw_b$ also belongs to ${\cal T}_B$.
    %An analogous statement holds if the points are colored red. 
\end{ourclaim}

\begin{proof}
Since $v_aw_a$ and $v_bw_b$ are contained in $\bluemst$, the new graph is also a spanning tree of $B$. Since $w_aw_b$ has the same length as $v_av_b$, the graph has minimum weight. \qed
\end{proof}

With Claims~\ref{cl:edges-crossing-half-integer-vertical-lines-first}, \ref{cl:edges-in-sections-horizontal-or-vertical} and~\ref{cl:choice-horizontal-edges} in hand, we can describe the structure of the MSTs. 
As already mentioned, if a column is monochromatic, the associated edge is present in any MST for the corresponding color. Inside every part, by Claims~\ref{cl:edges-crossing-half-integer-vertical-lines-first} and~\ref{cl:edges-in-sections-horizontal-or-vertical}, the additional edges are horizontal edges between two neighboring monochromatic columns of the same color (here ``neighboring" means neighboring \emph{if we ignore the columns of the other color}), and (possibly with the exception of $P_0$) two horizontal edges connecting the rainbow column of the section to the first monochromatic column of each color.  Finally, there is only one edge in each tree that connects two neighboring parts, and this edge is either horizontal (see $P_0$ and $P_1$ in Figure~\ref{fig-mstexample}) or diagonal (see $P_2$ and $P_3$ in Figure~\ref{fig-mstexample}).
%(if one of the section contains a single column). See Figure~\ref{fig-mstexample} for an illustration. 

\begin{figure}
    \centering
    \includegraphics[page=1]{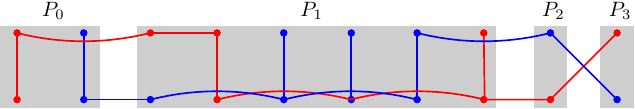}
    \caption{Possible $\bluemst\in {\cal T}_B$ and $\redmst\in {\cal T}_R$ in the proof of Lemma~\ref{lem-nonUniqueMST}.}
    \label{fig-mstexample}
\end{figure}

We now describe our choice of $\bluemst\in {\cal T}_B$ and $\redmst\in {\cal T}_R$. Suppose first that $P_j$ consists of more than one column. We describe how to select the horizontal edges of the MSTs. Since we only need to assign to $P_j$ non-vertical crossings, we only analyze horizontal edges longer than 1.5 (shorter horizontal edges do not participate in any crossing) and diagonal edges. 

If $j \neq 0$, the leftmost horizontal edge longer than 1.5 with both endpoints in $P_j$ (here left/right is decided according to the leftmost endpoint of the edge) connects the rainbow column  to the first monochromatic column in $P_j$ (if any) of one of the colors. Let us assume that this color is blue. Then this edge is present in any $\bluemst\in {\cal T}_B$. If $j=0$, the leftmost column of the part might not be rainbow, in which case, by Claim~\ref{cl:choice-horizontal-edges}, there might be two choices for the leftmost horizontal edge longer than 1.5 with both endpoints in $P_j$. We again assume that this edge is blue and  arbitrarily choose one of them. Let us now look at the next horizontal edge (i.e., second leftmost) with both endpoints in $P_j$ and longer than 1.5. This is a red edge connecting two red columns. By Claim~\ref{cl:choice-horizontal-edges}, there are two choices for this edge, so we can choose it so that it has the endpoints in the row opposite to that of the previous analyzed edge (i.e., the previous horizontal edge longer than 1.5 with both endpoints in $P_j$). In this way, the two edges do not cross. See the ``long" horizontal edges of $P_1$ in Figure~\ref{fig-mstsgoodchoice}.
We continue in this way, and at each step we choose an edge connecting points in the row opposite to that of the previously chosen edge. Eventually, we reach the end of part $P_j$, where we need to connect a point in the last monochromatic column of each color to the leftmost (i.e., the rainbow) column of the next part $P_{j+1}$. These edges are present in every $\bluemst\in {\cal T}_B$ and $\redmst\in {\cal T}_R$, and one of them is a horizontal edge of length at least $1.5$. This edge might participate in a non-vertical crossing, in which case the crossing is assigned to $P_j$; see Figure~\ref{fig-mstsgoodchoice}.  

%We now describe our choice of $\bluemst\in {\cal T}_B$ and $\redmst\in {\cal T}_R$. Note that, for each part $P_j$ consisting of more than one column, the leftmost edge in $\bluemst\cup\redmst$ with both endpoints in $P_j$ which is longer than 1.5 is the edge connecting the rainbow column to the first monochromatic column in one of the colors (assume this color is blue) and this edge is present in any choice of $\bluemst\cup\redmst$. The only exception is $j=0$ which does not contain a rainbow column and therefore we can choose this edge arbitrarily by Claim \ref{cl:choice-horizontal-edges}, in either case this edge is horizontal by Claim \ref{cl:edges-in-sections-horizontal-or-vertical}. Then, the next edge in $\bluemst\cup\redmst$ with both endpoints in $P_j$ and longer than 1.5 is going to be a red edge connecting two red columns, by Claim \ref{cl:choice-horizontal-edges}, we can choose this edge to have both endpoints in the opposite row from the previous one (and therefore not cross it). 
%We can continue this pattern and at each step choose an edge connecting points in a different row from the previous one, see the Figure \ref{fig-mstsgoodchoice}. Eventually, we reach the end of part $P_j$, where we need to connect a point in the last monochromatic column of each color to the leftmost (rainbow) column of the next part $P_{j+1}$, and this edge is present in any choice of $\bluemst\cup\redmst$. If this edge participates in any non-vertical crossing, we assign this crossing to $P_j$.  

\begin{figure}
    \centering
    \includegraphics[page=2]{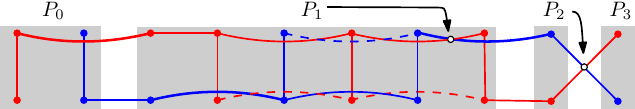}
    \caption{The choice of $\bluemst\in {\cal T}_B$ and $\redmst \in {\cal T}_B$ in the proof of Lemma \ref{lem-nonUniqueMST}. The arrows connect the parts to the non-vertical crossings assigned to them.}
    \label{fig-mstsgoodchoice}
\end{figure}

Finally, suppose that $P_j$ consists of just one column $C_i$. In this case, as already mentioned, in the convex hull of $C_i\cup C_{i+1}$ there might be at most one crossing, created by two diagonal edges. This crossing, if it exists, is assigned to~$P_j$. %Note that by Claim \ref{cl:edges-crossing-half-integer-vertical-lines}, we know that there cannot be more than one such crossing. 

In conclusion, given a coloring of $P$, it is possible to construct $\bluemst\in {\cal T}_B$ and $\redmst\in {\cal T}_R$ such that the trees cross at most $n/2-1$ times. \qed
%Therefore, we have successfully defined the choice of the MSTs and the assignment of the crossings to parts, finishing the proof. 
\end{proof}

We now prove Theorem~\ref{thm-grids}, which tells us that, in the generic setting, perturbations of the $2\times \frac{n}{2}$ grid cannot give us a matching upper bound for Theorem~\ref{thm-convex}. The proof of Theorem~\ref{thm-grids} is split into two parts, corresponding to the two lemmas below.
%This result follows immediately from Lemmas \ref{lem-n2-grid-tight} and \ref{lem-n2-grid-tight-2}, which we state and prove below.

In the proofs of Lemmas~\ref{lem-n2-grid-tight} and~\ref{lem-n2-grid-tight-2}, we use the following terminology: If a (perturbed) grid is colored in red and blue, we call the set of columns between two rainbow columns (in the left-to-right order) a \emph{section}. Note that a section might be empty.
 
\begin{lemma}\label{lem-n2-grid-tight}
    Suppose that $n\ge2$ is even. Let $P$ be any perturbation of a $2\times \frac n2$ grid that yields a generic convex point set. Then $\cross(P)\ge \lfloor\frac{5}{8}(n-2)\rfloor$.
    %Suppose that $n=8k+2$. Then any perturbation of the $2\times\frac{n}{2}$ grid that gives a convex point set of pairwise different distances admits a coloring such that the MSTs of the two colors cross at least $\frac{5}{8}(n-2)$ times.
\end{lemma}
\begin{proof}
A coloring of $P$ yielding exactly $\lfloor\frac{5}{8}(n-2)\rfloor$ crossings is defined as follows.
Suppose $n=8k+2+2d$, where $k\ge0$ is an integer and $d\in\{0,1,2,3\}$. The columns $C_1,C_5,C_9,\dots,C_{4k+1}$ have $v_{4i+1}$ colored blue and $w_{4i+1}$ colored red, so they are rainbow.
In each row, the colors alternate to the right of the column $C_{4k+1}$. This makes the $d$ columns to the right of $C_{4k+1}$ rainbow. 
% See Figure~\ref{fig-n2grid-tight-0}. 

% \pavel{ I suggest a new figure fig-n2grid-tight-0.pdf here (it is already prepared in the latex file; just needs to be drawn. :-) The figure would show an example with $n=22$, say. Yes, I like figures. :-)}

%      \begin{figure}[h!]
%         \centering
%            \includegraphics[page=3]{Figures/longproof.pdf}
%            \caption{Possible coloring of the grid in the proof of Lemma~\ref{lem-n2-grid-tight} for $n=22$ (and $k=2$, $d=2$).}
%           \label{fig-n2grid-tight-0}
%        \end{figure}

         \begin{figure}[h!]
            \centering
            \includegraphics{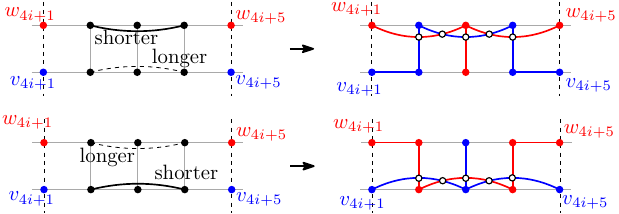}
            \caption{Coloring  a non-empty section in the proof of Lemma~\ref{lem-n2-grid-tight} (two cases).}
            \label{fig-n2grid-tight}
        \end{figure}

All columns that are not yet colored will be monochromatic, with their colors determined as follows.
Within a section delimited by the rainbow columns $C_{4i+1}$ and $C_{4i+5}$, we compare the lengths of $w_{4i+2}w_{4i+4}$ and $v_{4i+2}v_{4i+2}$. If $|w_{4i+2}w_{4i+4}|<|v_{4i+2}v_{4i+4}|$, we color the middle column $C_{4i+3}$ in red and the remaining two columns in blue; see Figure~\ref{fig-n2grid-tight} (top). Otherwise, we color the middle column $C_{4i+3}$ in blue and the rest in red; see Figure~\ref{fig-n2grid-tight} (bottom). 

%Then, the MSTs within each section are as shown in Figure~\ref{fig-n2grid-tight}.

Each non-empty section contributes 5 crossings, and each empty section contributes 1 crossing. Thus, in total the MSTs cross $5k+d=5k+\lfloor \frac58(2d)\rfloor =\lfloor \frac58(8k+2d)\rfloor=\lfloor \frac58(n-2)\rfloor $ times; see Figure~\ref{fig-n2grid-tight}. \qed
\end{proof}

Lemma~\ref{lem-n2-grid-tight} is best possible:
\begin{restatable}{lemma}{nTwoGridTightTwo}
\label{lem-n2-grid-tight-2}
Suppose that $n \ge 2$ is even. There exists a generic perturbation $P_0$ of the $2\times\frac n2$ grid in convex position such that $\cross(P_0) \leq \lfloor\frac{5}{8}(n-2)\rfloor$.
\end{restatable}

\begin{proof}
We obtain a perturbation of a $2\times \frac{n}{2}$ grid satisfying Lemma~\ref{lem-n2-grid-tight-2} by following three perturbation steps described below. Although we give specific small ``perturbation'' constants, it is apparent from the construction and from the proofs that the final perturbation can be arbitrarily small. That is, all the points can be arbitrarily close to their original grid positions.
Here are the three perturbation steps:
\begin{enumerate}
\item We first slightly increase the $y$-coordinate of the points of type $w_i$ from 1 to 1.01.
\item We next slightly decrease the $x$-coordinates of the points $v_2,v_4,v_6\ldots$ and $w_1,w_3,w_5\ldots$ by scaling them by a factor slightly smaller than $1$, namely by the factor $1-1/(100n)$, so that, for every $i=1,2,\ldots,n/2-1$, we still have $|v_iv_{i+1}|<1.01$ and $|w_iw_{i+1}|<1.01$. We observe that, after this step, $|w_{2j+1}w_{2j+3}|<|v_{2j+1}v_{2j+3}|$ for $j=0,1,\ldots$, since $|w_{2j+1}w_{2j+3}|<1$ and $|v_{2j+1}v_{2j+3}|=1$. Analogously, $|v_{2j}v_{2j+2}|<|w_{2j}w_{2j+2}|$ for $j=1,2,\ldots$
\item Finally, we infinitesimally perturb the points so that they are in convex position and all distances between pairs of points are distinct, while still maintaining:
\begin{itemize}
\item[(a)] $\max_{i \in 1,2,\ldots,n/2-1} \{|v_iv_{i+1}|,|w_iw_{i+1}| \}<\min_{i \in 1,2,\ldots,n/2} \{ |v_iw_{i}|\}$,
\item[(b)] $|w_{2j+1}w_{2j+3}|<|v_{2j+1}v_{2j+3}|$ for $j=0,1,\ldots,$
\item[(c)] $|v_{2j}v_{2j+2}|<|w_{2j}w_{2j+2}|$ for $j=1,2,\ldots$, and
\item[(d)] each point lies in the $0.02$-neighborhood of its initial grid position.
\end{itemize}
\end{enumerate}

 We call edges of type $v_iv_j$ or $w_iw_j$ \emph{horizontal}, and edges of type $v_iw_{i}$ \emph{vertical}. The remaining edges are called \emph{diagonal}. The \emph{width} of a horizontal edge $v_iv_j$ or $w_iw_j$ is $|i-j|$. Notice that $|v_iv_j|$ or $|w_iw_j|$ is not necessarily $|i-j|$, but it is very close to it. We observe that property (a) of the point set implies that any vertical edge is longer than any horizontal edge of width 1.

For points $v_i$ and $w_i$, we denote the set of points $\{v_j,w_j:\ j<i\}$ the points \emph{to their left} and the set of points $\{v_j,w_j:\ j>i\}$ the points \emph{to their right}. We observe that, due to the perturbation, $v_i$ might also be to the left or right of $w_i$, but we do not include each other in the left and right sets.

Let us fix a coloring of the point set.

%\pavel{I strengthened the lemma below. The original statement easily follows. (The original statemant was: Suppose that $v_i$ has color blue. Then, in the blue MST, $v_i$ is adjacent to at most one blue point to each side, namely, the closest blue point on that side.An analogous statement holds for $w_i$.}

%%%%%%%%%%%%%%%%%%%%%%%%%%%%%%%%%%%%%%%%%%%%%%%%%%%%%%%%%%%%%
\begin{restatable}{ourclaim}{edgescrossinghalfint}\label{cl:edges-crossing-half-integer-vertical-lines} Let $i \in \{1,\dots,n/2-1\}$. Suppose that two distinct edges $e,f$ of  $T_B$ intersect the vertical line $x = i+1/2$.
Then all four points $v_i, v_{i+1}, w_i, w_{i+1}$ are blue and the edges $e,f$ are the edges $v_i v_{i+1}$ and $w_i w_{i+1}$. 
An analogous statement holds for $\redmst$.
\end{restatable}

\begin{proof}
Identical to the proof of Claim \ref{cl:edges-crossing-half-integer-vertical-lines-first}. \qed
\end{proof}

%%%%%%%%%%%%%%%%%%%%%%%%%%%%%%%%%%%%%
\begin{ourclaim} \label{cl:hor-edges}
If $v_i$ and $v_{i+1}$ have the same color, the edge $v_iv_{i+1}$ belongs to the minimum spanning tree of that color. An analogous statement holds for the points $w_i$.
\end{ourclaim}

\begin{proof}
Suppose, for the sake of contradiction, that $v_iv_{i+1}$ does not belong to the minimum spanning tree of that color. Then the path in the minimum spanning tree from $v_i$ to $v_{i+1}$ contains a vertical edge, a diagonal edge or a horizontal edge of width 2 or more. Edges of all of these types are longer than $|v_iv_{i+1}|$, so we reach a contradiction. \qed
\end{proof}

We call a column $v_i,w_i$ \emph{rainbow} if $v_i$ and $w_i$ have different color; otherwise, we call it \emph{monochromatic}.

%%%%%%%%%%%%%%%%%%%%%%%%%%%%%%%%%%%%%%%%%%%%%%%%%%%%%%%%%%%%%%
\begin{ourclaim} \label{cl:first/last-column-change-to-rainbow}
If $C_1$ is monochromatic, it is possible to change the color of one of the points of the column in such a way that the number of crossings between the two minimum spanning trees does not decrease.
An analogous statement holds for the last column $C_{n/2}$.
\end{ourclaim}

\begin{proof}
Suppose that the points $v_1,w_1$ are blue. 
If no edge of the blue MST $\bluemst$ intersects the vertical line $x=1.5$ then there is no crossing between $\bluemst$ and $\redmst$, and the claim trivially holds. Suppose now that exactly one edge of $\bluemst$ intersects the line $x=1.5$. Then $\bluemst$ contains the edge $v_1w_1$ and one of the vertices $v_1,w_1$, further denoted by $\ell,$ is a leaf of $\bluemst$. We claim that changing the color of $\ell$ from blue to red gives the claim in this case. Indeed, considering Kruskal's algorithm, we observe that (i) the new blue MST is obtained from the original one by removing the edge $v_1w_1$, which did not participate on any crossing between $\bluemst$ and $\redmst$, and, due to Claim~\ref{cl:edges-crossing-half-integer-vertical-lines} (ii) the new red MST is obtained from the original one by adding an edge incident to $\ell$. The claim therefore holds also in this case.

Due to Claim~\ref{cl:edges-crossing-half-integer-vertical-lines}, it remains to verify the case when all four points $v_1,v_2,w_1,w_2$ are blue and $\bluemst$ contains the two edges $v_1v_2$ and $w_1w_2$. If we change the color of $v_1$ from blue to red then, due to Kruskal's algorithm, (1) the new $\bluemst$ is obtained from the original one by removing the edge $v_1v_2$ and eventually replacing an edge $v_1w_1$ (if it is in $\bluemst$) by another edge, and (2) the new $\redmst$ is obtained from the original one by adding an edge incident to $v_1$. Since the edges $v_1v_2$ and $v_1w_1$ cannot participate in any crossing, the claim is proved. \qed
\end{proof}

Consequently, we can assume that $C_1$ and $C_{n/2}$ are rainbow.

We consider all columns that are rainbow, recall that a \emph{section} of a (preturbed) grid is the subset of points strictly between two consecutive rainbow columns (in the left-to-right order). In particular, if any two consecutive columns are both rainbow, the empty set of points between them is considered as an \emph{empty section}.
For a section $S$ delimited by the rainbow columns $C_i$ and $C_j$, we define $S^+:=S\cup C_i\cup C_j$. See Figure~\ref{fig:section}. 

\begin{figure}[h!]
    \centering
    \includegraphics[page=2]{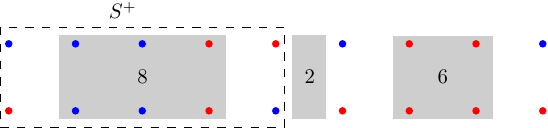}
    \caption{Sections in a grid: each section is marked by a grey rectangle and the weight of the section is indicated.}
    \label{fig:section}
\end{figure}

%Thus, if $C_{1}=C_{i_1},C_{i_2},\dots,C_{i_k}=C_{n/2}$ are all the rainbow columns (in the left-to-right-order) then there are $k-1$ sections, and if $S=C_{i_j}the $j$-th section is
In the rest of the proof, we consider various types of intervals of consecutive monochromatic columns.
Let ${\cal C}=\{C_i,C_{i+1},\dots,C_{i+j-1}\}$ be such an interval of $j$ consecutive monochromatic columns, and let $r({\cal C)}\in\{0,1,2\}$ be the number of rainbow columns in the pair $\{C_{i-1},C_{i+j}\}$. Then the \emph{weight} $w({\cal C})$ of ${\cal C}$ is $2j+r({\cal C})$. In particular, each empty section (delimited by two consecutive rainbow columns) has weight $2$. 

%%%%%%%%%%%%%%%%%%%%%%%%%%%%%%%%%%%%%%%%%%%%%%%%%%%
\begin{ourclaim} \label{cl:total-weight-of-sections}
 The total weight of all the sections, including the empty ones, is $n-2$. (Recall that we assume that $C_1$ and $C_{n/2}$ are rainbow.)

%% (ii) Suppose that a section $S$ is partitioned into intervals of consecutive columns. Then the total weight of these intervals is equal to the weight of $S$.
\end{ourclaim}

\begin{proof}
The weight of a section between two rainbow columns $C_i,C_j$ is $2(j-i-1)+2=2(j-i)$.
Thus, if $1=i_1,i_2,\dots,i_k=n/2$ are the indices of the rainbow columns then the total weight of the sections is $\sum_{j=1}^{k-1} 2(i_{j+1}-i_j)=2(i_k-i_1)=2(n/2-1)=n-2$. \qed
% Each monochromatic column contributes $2$ to the weight of its section.
% Each rainbow column contributes $1$ to each of the two sections just before and just after it (if they exist).
% Thus, the rainbow columns $v_1,w_1$ and $v_{n/2},w_{n/2}$ contribute $1$ to the first (resp. last) section, and every other rainbow column contributes $1+1=2$ in total.
% The claim follows.
\end{proof}

\begin{ourclaim}\label{cl:interval-weights}
Suppose that a section $S$ is partitioned into intervals of consecutive columns. Then the total weight of these intervals is equal to the weight of $S$.
\end{ourclaim}

\begin{proof}
By the definition of weight, both quantities are equal to two plus the number of (monochromatic) columns in $S$. \qed
\end{proof}

%%%%%%%%%%%%%%%%
\begin{ourclaim}\label{cl:crossings-in-S+}
Let $S$ be a section of weight $w(S)$. %between two rainbow columns $C_i,C_j,i<j$.
%(Thus, $S^+=S\cup C_i\cup C_j$.)
Then $\conv(S^+)$ contains at most $(5/8)w(S)$ crossings between the two MSTs. (Recall that $S^+$ denotes the section $S$ extended by the two rainbow columns delimiting it.)
\end{ourclaim}

The union $\bigcup\conv(S^+)$ taken over all sections $S$ covers the convex hull of our point set. Therefore, Claims~\ref{cl:total-weight-of-sections} and~\ref{cl:crossings-in-S+} immediately give Lemma~\ref{lem-n2-grid-tight-2}.

We now show claims needed for the proof of Claim~\ref{cl:crossings-in-S+}.

%The idea of the proof is to show that there within the convex hull of a section (including its bounding rainbow columns) of weight $w$, the number of 

A \emph{block} is a maximal set of consecutive monochromatic columns of the same color. We observe that a non-empty section is partitioned into blocks whose colors alternate. A block is \emph{blue} or \emph{red}, according to the color of its columns.

Recall that we say that a crossing is \emph{vertical} if one of the two edges involved is of type $v_iw_i$; otherwise, it is \emph{non-vertical}.

%%%%%%%%%%%%%%%%
\begin{ourclaim}\label{cl:exactly-one-crossing-in-every-block}
For every block, there is exactly one crossing between $\bluemst$ and $\redmst$ lying in the convex hull of the points forming the block. This crossing is a vertical crossing that involves a vertical edge within the block.
\end{ourclaim}

\begin{proof}
Suppose that the block is formed by columns $v_j,w_j;v_{j+1},w_{j+1};\ldots;v_{j+\ell},w_{j+\ell}$ of color blue. By Claim~\ref{cl:hor-edges}, all the edges along the paths $v_jv_{j+1}\ldots v_{j+\ell}$ and $w_jw_{j+1}\ldots w_{j+\ell}$ belong to $\bluemst$. This already implies that there is at most one vertical edge within the block. By Claim~\ref{cl:edges-crossing-half-integer-vertical-lines}, there are no diagonal blue edges within the block. We next prove that there is exactly one vertical edge: Suppose, for the sake of contradiction, that there is no vertical edge within the block. Then, since in each side of the block there is a pair such that at least one of the two points is red, the path from $v_j$ to $w_j$ in the blue MST cannot contain only horizontal edges of width one and vertical edges outside the block. Thus, this path contains a diagonal edge or a horizontal edge of width 2 or more. Since both types of edges are longer than $v_jw_j$, we obtain a contradiction.

In consequence, there is exactly one vertical edge within the block. Since there are red points both to the left and right of the block, there is an edge of $\redmst$ connecting red points from both sides. Due to Claim~\ref{cl:edges-crossing-half-integer-vertical-lines}, this edge is unique. This edge intersects the vertical blue edge within the block creating a crossing.  There is no other crossing between $\bluemst$ and $\redmst$ lying in the convex hull of the points forming the block, since the unique red edge intersecting the block intersects no edge of the two paths $v_jv_{j+1}\ldots v_{j+\ell}$ and $w_jw_{j+1}\ldots w_{j+\ell}$ and there is no additional edge of $\bluemst$ intersecting the block due to Claim~\ref{cl:edges-crossing-half-integer-vertical-lines}. \qed
 \end{proof}

\begin{ourclaim} \label{cl:empty-section-1}
For every empty section $S$, there is at most one crossing in the interior of the convex hull of $S^+$, and no crossing on the boundary.
\end{ourclaim}

\begin{proof}
Direct consequence of Claim~\ref{cl:edges-crossing-half-integer-vertical-lines}. \qed
\end{proof}

\begin{ourclaim} \label{cl:mono-rain}
Given two consecutive columns, one of them being monochromatic and the other rainbow, there is no crossing lying in the interior of the convex hull of the two columns.
\end{ourclaim}

\begin{proof}
Direct consequence of Claims~\ref{cl:edges-crossing-half-integer-vertical-lines} and~\ref{cl:hor-edges}. \qed
\end{proof}

\begin{ourclaim} \label{cl:two-columns}
For any two consecutive columns, there is at most one crossing lying in the interior of the convex hull of the two columns.
\end{ourclaim}

\begin{proof}
Direct consequence of Claim~\ref{cl:edges-crossing-half-integer-vertical-lines}. \qed
\end{proof}

A block is \emph{short} if it contains only one column, and it is \emph{long} otherwise. 

The next two key claims show that, in some situations, the number of crossings is even more restricted:

\begin{ourclaim} \label{cl:no111}
Consider four consecutive monochromatic columns $C_i,\dots,C_{i+3}$ such that $C_{i+1}$ and $C_{i+2}$ are short blocks. Then there is no crossing lying in the interior of $\CH(C_{i+1}\cup C_{i+2})$.
\end{ourclaim}

\begin{figure}[ht]
    \centering
    \includegraphics[page=1]{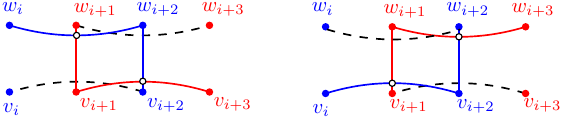}
    \caption{Two possible situations in Claim~\ref{cl:no111}, longer edges are drawn with dashed lines.} 
    \label{fig-no111}
\end{figure}

\begin{proof}
Due to condition (c) in step 3. of the construction of the perturbation of the grid, we have either $|v_{i+2}v_i|<|w_{i+2}w_i|$ and $|v_{i+3}v_{i+1}|>|w_{i+3}w_{i+1}|$ (if $i$ is even),
or $|v_{i+2}v_i|>|w_{i+2}w_i|$ and $|v_{i+3}v_{i+1}|<|w_{i+3}w_{i+1}|$ (otherwise). In the first case the only two edges of the red and blue MSTs intersecting the vertical line $x=i+1/2$ are $v_{i+2}v_i$ and $w_{i+3}w_{i+1}$. In the latter case they are $v_{i+3}v_{i+1}$ and $w_{i}w_{i+2}$. In either case, the two edges do not cross; see Figure~\ref{fig-no111}.  The claim follows. \qed
\end{proof}

\begin{ourclaim} \label{cl:no12111}
Consider six consecutive monochromatic columns $C_i,\dots,C_{i+5}$ such that $C_{i+1}$ and $C_{i+4}$ are short blocks and $C_{i+2}\cup C_{i+3}$ is a long block between them. Then either there is no crossing in the interior of $\CH(C_{i+1}\cup C_{i+2})$ or there is no crossing in the interior of $\CH(C_{i+3}\cup C_{i+4})$. 
\end{ourclaim}

\begin{figure}[h]
    \centering
    \includegraphics[width=\textwidth, page=3]{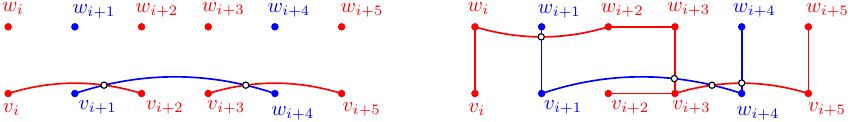}
    \caption{The two non-vertical crossings on the left cannot be present simultaneously by Claim~\ref{cl:no12111}. On the right we can see one possible situation in Claim~\ref{cl:no12111}.}
    \label{fig-no121}
\end{figure}

\begin{proof}
Without loss of generality, suppose that the column $C_i$ is red.
If $i$ is even then $|v_{i}v_{i+2}|<|w_{i}w_{i+2}|$ and $|v_{i+3}v_{i+5}|>|w_{i+3}w_{i+5}|$, and therefore the edges $v_{i}v_{i+2}$ and $w_{i+3}w_{i+5}$ are the two unique edges of $\redmst$ crossing the two convex hulls of our interest. (Again, we have used Claim~\ref{cl:edges-crossing-half-integer-vertical-lines}.) The unique edge of $\bluemst$ crossing the two convex hulls of our interest is either $v_{i+1}v_{i+4}$ or $w_{i+1}w_{i+4}$. In the first case there is no crossing in the right of the two convex hulls, and in the latter case there is no crossing in the left one.

If $i$ is odd then the signs in the two inequalities considered in the even case are flipped, and the rest of the argument is analogous as in the even case. \qed
\end{proof}

Blocks are grouped into \emph{clusters} as follows: A \emph{big cluster} is a maximal interval of consecutive long blocks. A \emph{small cluster} is a maximal interval of consecutive short blocks. In some particular cases, blocks are grouped differently: If we have the sequence: rainbow column, short block, long block with 2 columns, short block, rather than ``grouping" the first two blocks into a short cluster and a long cluster, we group them together into a so-called \emph{medium cluster}. If we have the sequence: short block, long block with 2 columns, short block, rainbow column, we group the second and third block together into a \emph{medium cluster}. If we have the two situations happening at the same time, that is, rainbow column, short block, long block with 2 columns, short block, rainbow column, we group the first two blocks into a \emph{medium cluster} and we leave the third block as a small cluster. 

%\maria{I am assuming, Pavel, that you revised the proof until here.}
%\pavel{yes, to some extend}

For a fixed cluster, we call the crossings lying inside the convex hull of the points in the cluster \emph{interior}. We call \emph{boundary} the crossings lying in the interior of the convex hull of the following four points: either the left-most column in the cluster and the column to its left; or the right-most column in the cluster and the column to its right. 

To a big or medium cluster, we assign all interior crossings and all boundary crossings. To a small cluster, we assign all interior crossings.

Let us now consider an empty section $S$. To this section, we assign the unique crossing (see Claim~\ref{cl:empty-section-1}) lying in the convex hull of $S^+$ (if it exists).

\begin{ourclaim}
Each crossing is assigned to exactly one cluster or one empty section.
\end{ourclaim}
\begin{proof}
As already mentioned, the union $\bigcup\conv(S^+)$ taken over all sections $S$ covers the convex hull of our point set. 

Given an empty section $S$, by Claim~\ref{cl:empty-section-1} any crossing lying in the convex hull of $S^+$ lies in the interior of this convex hull. Therefore, it is only assigned to $S$.

Suppose next that section $S=\{C_i,\ldots,C_{i+j-1}\}$ is not empty. By Claim~\ref{cl:mono-rain}, there is no crossing lying in the interior of the convex hulls of $C_{i-1}$ and $C_i$, and of $C_{i+j-1}$ and $C_{i+j}$. The blocks in $S$ are grouped into clusters. For each cluster, crossings inside the convex hull of the points in the cluster are assigned to that cluster. In the left to right sequence of clusters in $S$, there is always an alternation between big or medium clusters and small clusters. This implies that crossings lying ``between'' two consecutive clusters (i.e., in the interior of the convex hull of two columns corresponding to the right-most column of one cluster and the left-most column of the next cluster) are assigned precisely to one of the two clusters (the big or medium one). \qed
\end{proof}

\begin{ourclaim}\label{cl:cluster-weight}
A cluster of weight $w$ is assigned at most $\frac{5w}{8}$ crossings.
\end{ourclaim}

\begin{proof}
We divide the proof into two main cases.

Let us first assume that the cluster does not contain the first or the last block of the section. 

If the cluster is small, it has only assigned the interior crossings. By Claim~\ref{cl:exactly-one-crossing-in-every-block}, there are $w/2$ vertical crossings lying in the convex hulls of the individual columns forming the cluster. By Claim~\ref{cl:no111}, there are no other crossings lying on the convex hull of the cluster. 

The cluster cannot be medium because, by definition, a medium cluster always contains the first or last block of the section.

Finally, suppose that the cluster is big and is made of long blocks with weight $w_1,w_2,\ldots,w_{\ell}$. The number of interior crossings assigned to it is $l$ vertical crossings (by  Claim~\ref{cl:exactly-one-crossing-in-every-block}) and at most $l-1$ crossings ``between'' consecutive blocks (by Claim~\ref{cl:two-columns}). By Claim~\ref{cl:two-columns}, the number of boundary crossings assigned to it is at most $2$. Hence, we get at most $2\ell +1$ crossings. Using that $\ell\leq w/4$, we get $\frac{w}{2}+1$ crossings, which is at most $\frac{5w}{8}$ if $w\geq 8$, so if there are at least two long blocks in the cluster. If there is one long block in the cluster, we have at most 3 crossings. If $w\geq 6$, this is smaller than $\frac{5w}{8}$. In consequence, it only remains to look at the case where there is one long block with $w=4$. 
%If it is the last block of the section, there is no boundary crossing to the right of the block, so it gets assigned only 2 crossings, which is fine. 
Since, by assumption, this block is not the first or the last block of the section, there is a short block to its left and a short block to its right. By definition, this is not a middle cluster, so this short block to the left is not the first block of the section, and the short block to the right is not the last block of the section. Thus, we can apply Claim~\ref{cl:no12111} to see that the number of boundary crossings assigned to the block is at most 1, and the total number of crossings is 2. Since $w=4$, this is smaller than $\frac{5w}{8}$.

The second main case is when the cluster contains the first or last block of the section.

If the cluster is small, it has only assigned the interior crossings. As before, there are $w/2$ vertical crossings lying in the convex hulls of the individual columns forming the cluster. We next analyze the other crossings lying on the convex hull of the cluster. We consider two subcases.

Suppose first that the cluster contains one among the first and last blocks of the section, say the first. If the cluster contains one column, its width is 3, the number of crossings is 1 (only a vertical one), and thus the bound holds. If it contains two columns, the number of crossings is at most 3 (2 vertical and at most one between the two columns) and the width is 5, so the bound also holds. If the cluster contains three or more columns, there are $\frac{w-1}{2}$ vertical crossings and possibly a non-vertical crossing between the first and second columns (but not between any other pair of consecutive columns due to Claim~\ref{cl:no111} and the fact that the last column of the cluster is not the last block of the section). Thus, we have at most $\frac{w-1}{2}+1$ crossings, which is smaller than $\frac{5w}{8}$ because $w\geq 7$.

Next, suppose that the cluster contains both the first and last blocks of the section. If the cluster contains one column, its width is 4, the number of crossings is 1, and the bound holds. If it contains two columns, the number of crossings is at most 3 (as in the previous case) and the width is 6, so the bound also holds. If the cluster contains three or more columns, there are $\frac{w-2}{2}$ vertical crossings, possibly a non-vertical crossing between the first and second columns and possibly a non-vertical crossing between the last two columns. Thus, we have at most $\frac{w-2}{2}+2$ crossings, which is smaller than or equal to $\frac{5w}{8}$ because $w\geq 8$. We observe that equality is achieved precisely at $w=8$.

If the cluster is medium, the number of interior crossings assigned to it is $2$ vertical crossings and at most $1$ crossing ``between'' consecutive blocks. Since the cluster contains the first or last block of the section, one of the two potential boundary crossings does not occur thanks to Claim~\ref{cl:mono-rain}. Hence, the total number of crossings is at most $4$. Since the weight of the cluster is 7, the bound holds.

Finally, suppose that the cluster is big and is made of long blocks with weight $w_1,w_2,\ldots,w_{\ell}$. As in the first case, the number of interior crossings assigned to it is $l$ vertical crossings and at most $l-1$ crossings ``between'' consecutive blocks. As in the previous case, the number of boundary crossings is at most 1. Hence, we get at most $2\ell$ crossings. Using that $\ell\leq w/4$, we get at most $\frac{w}{2}$ crossings. \qed
\end{proof}

\begin{ourclaim}\label{cl:empty-section-1}
An empty section $S$ is assigned less than $\frac{5}{8}w(S)$ crossings.
\end{ourclaim}

\begin{proof}
The proof trivially follows from the facts that $w(S)=2$ and $S$ is assigned at most one crossing. \qed
\end{proof}

Now, Claims~\ref{cl:total-weight-of-sections},~\ref{cl:interval-weights},~\ref{cl:cluster-weight},~\ref{cl:empty-section-1} imply Lemma~\ref{lem-n2-grid-tight-2}. \qed

\end{proof}

\subsection{Random Point Sets}
In this section, we consider the case where $P$ is a set of $n$ random points uniformly distributed in $[0,1]^2$.
%, and prove Theorem \ref{thm-random}.

\random* 
\begin{proof} 
We divide the square $[0,1]^2$ into square cells of area $\frac{4}{n}$ each. We further divide each such cell into $121$ equal subcells of area $\frac{4}{121n}$. We call a cell \emph{good} if all subcells are empty except for the four subcells adjacent to the centermost subcell, which contain a single point each (see Figure~\ref{fig-random}). Note that the distance between any two points inside a good cell is smaller than the distance between a point in a good cell and a point in any other cell. 

For a single cell $C$, we have $\mathbb{P}[C \text{ is good}]= \binom{n}{4}\left(\frac{4}{121n}\right)^4 \left( 1-\frac{4}{n} \right)^{n-4}.$ For $n>4$, this is decreasing and converges to $K>0$. Now, let $A_C$ be the event that there is a crossing between $\redmst$ and $\bluemst$ inside a cell $C$. Then we have  $\mathbb{P}[A_C|C \text{ is good}] = \frac{1}{8}$, since two out of the $16$ possible colorings of points in a good cell $C$ induce a crossing, and this crossing occurs independently from the coloring of the points outside $C$. Finally, the result follows by the linearity of expectation. \qed \hfill 

    \begin{figure}[h]
        \centering
        \includegraphics[ page=5]{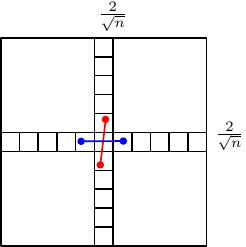}
        \caption{A good cell in the proof of Theorem~\ref{thm-random}.}
        \label{fig-random}
    \end{figure}
\end{proof}

%Combining Theorems~\ref{thm-genupperbound} and~\ref{thm-random} gives us the following as a corollary. 

%\begin{corollary}
%Let $P$ be a set of $n$ points uniformly distrbuted in $[0,1]^2$. If we color $P$ into two colors $P=R\cup B$ uniformly at random, then there is a constant $c>0$ such that $\mathbb{P}[\crossAB(R,B)\ge cn] \to 1$ as $n\to \infty$. 
%\end{corollary}

%\textcolor{orange}{Todor: I will prove it or delete it, might not be as easy as I thought}

\subsection{Dense Point Sets }

In this section, our goal is to prove Theorem~\ref{thm-dense}. 
%But before we proceed with the proof, we need some preparation. 

Let $P$ be a point set in the plane such that $\min_{a,b\in P}d(a,b)=1$ (we can always assume this up to rescaling). We define the \emph{minimum bounding box} of $P$ to be a minimum area square with integer-length edges that contains $P$. We further denote the area of the minimum bounding box of $P$ by $B(P)$. Then, we define the \emph{density} of $P$ as $\phi(P)=\frac{|P|}{B(P)}$. This notion of density is more useful for our purposes and is closely related to the original definition of $\alpha$-density. More formally, we have the following. 

\begin{lemma}\label{lem:densesetsaredense}
    Let $P$ be a generic $\alpha$-dense set of $n$ points in the plane for some $\alpha>1$. Then $\phi(P)\ge \frac{1}{\lceil\alpha \rceil^2}$.
\end{lemma}

\begin{proof}
    Since $P$ is $\alpha$-dense, we know that the diameter of $P$ is at most $\alpha\sqrt{n}$. Thus, $P$ is contained in a square with side length $\lceil\alpha\sqrt{n}\rceil$, which implies that $B(P)\le \lceil\alpha\rceil^2n$. This in turn gives $\phi(P)\ge \frac{1}{\lceil\alpha\rceil ^2}$. \qed
\end{proof}

%The following result was proven by Valtr [\cite{Valtr1992}, Proposition 4.10]. 

%\begin{lemma}[\cite{Valtr1992}, Proposition 4.10] \label{lem:pavel}
 %   For every $\alpha < \frac{\sqrt{2}\sqrt[4]{3}}{\sqrt{\pi}}$ there is a number $k(\alpha)\in \mathbb{N}$ such that there is no set of more than $k(\alpha)$ points in the plane with the minimum distance $1$ and diameter less than $\alpha\sqrt{n}$. 
%\end{lemma}

The following corollary can be obtained by an easy packing argument.

\begin{corollary}\label{cor:densitybounded}
    For any $\delta\ge 3$, there is a number $k(\delta)\in \mathbb{N}$ such that there is no set $P$ of more than $k(\delta)$ points in the plane with $\phi(P) = \delta$.
\end{corollary}

%Todor: I can add a proof but it's really just using the previous lemma for a suitable alpha. 

For a set of points $P$ and integer $k$ we say that $P$ can \emph{fill in} a $k\times k$ grid if we can place a (possibly rescaled) copy of $P$ inside the square $[0,k]^2$  in such a way that each subsquare of the grid $\{0,1,\dots,k\}\times \{0,1,\dots, k\}\subset \mathbb{Z}\times \mathbb{Z}$ contains at least one point of $P$ in its interior.

%For a set of points $P$ and integer $k$ we say that $P$ can \emph{fill in} a $k\times k$ grid if we can place a (possibly rescaled) copy of $P$ into $\{0,1,\dots,k\}\times \{0,1,\dots, k\}\subset \mathbb{Z}\times \mathbb{Z}$ in such a way that each subsquare of the grid contains at least one point of $P$ in its interior.

\begin{restatable}{lemma}{fullchessboard}
\label{lem:fullchessboard}
Let $P$ be a generic set of points in the plane. If there is a subset $P' \subseteq P$ that can fill in a $101 \times 101$ grid, there is a coloring of $P'$ such that $\redmst$ and $\bluemst$ cross at least once, regardless of how we color the points of $P$ outside the $101 \times 101$ grid.
\end{restatable}

\begin{proof}
    We rescale and translate $P$ so that each subsquare of $\{0,1,\dots,101\}\times \{0,1,\dots, 101\}$ contains at least one point of a (possibly rescaled) copy of $P'$ in its interior. The conclusion of the lemma will still hold for the original copy of $P$. 
    %Before starting the proof, $P$ might be rescaled in order for $P'$ to fill in the grid, but our argument still works on the original copy of $P$ (but we might need to use a grid with larger squares). 
    Let $Q\subset P'$ be the set of points in $P'$ contained in $\{40,41,42,\dots,61\}\times \{40,41,42,\dots,61\}$. That is, $Q$ is the set of points inside the ``innermost'' $21\times 21$ square of the $101\times 101$ grid. Further, let $Q'\subset Q$ be the set of points in $Q$ contained in $\{42,43,\dots,59\}\times \{42,43,\dots,59\}$. Lastly, let $q \in Q$ be a point of $Q$ inside $\{50,51\}\times \{50,51\}$, the innermost $1\times1$ subsquare of the grid. Now, let $T[q]$ be an inclusion-minimal subpath of $T_Q$ connecting $q$ with some point in $Q\setminus Q'$. We can assume that $T[q]$ terminates in a point inside the square $\{59,60\}\times\{x,x+1\}$ for some $x \in \{41,42,\dots,59\}$. Let $r_1 \in P \cap\{58,59\} \times \{40,41\}$ and $r_2\in P \cap \{58,59\} \times \{60,61\}$. Let us color $r_1,r_2$ red and all other points of $P'$ blue. Then, $\redmst$ consists of the single edge $r_1r_2$ and this edge crosses $T[q]$, see Figure~\ref{fig-fullchessboard}.
    
    \begin{figure}[ht]
    \centering
    \includegraphics[page = 6]{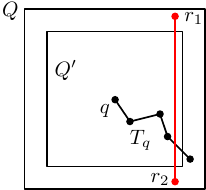}
    \caption{Obtaining a crossing in the proof of Lemma~\ref{lem:fullchessboard}}
    \label{fig-fullchessboard}
\end{figure}
    
    Further, distance between any two points in $Q$ is less than the distance between a point in $Q$ and a point in $P\setminus P'$ so this crossing is present however we color the points outside $P'$.  \qed
\end{proof}

In order to apply Lemma~\ref{lem:fullchessboard}, we need the following key observation.

\begin{lemma} \label{lem:densityimproving}
    Let $P$ be a generic set of $n$ points in the plane such that $\phi(P)=\delta$ and let $k\in \mathbb{N}$ be a positive integer. Then, $P$ fills in a $k\times k$ grid or there is a subset $Q\subset P$ such that $\phi(Q)\ge (1+\frac{1}{k^2-1})\delta$ and $|Q|\ge \frac{1}{k^2-1}|P|$.  
\end{lemma}

\begin{proof}
Assume that $P$ does not fill in a $k\times k$ grid. Let $B$ be the minimum bounding box of $P$. Then, if we divide $B$ into $k^2$ subsquares of area $\frac{B(P)}{k^2}$, by the pigeonhole principle we can find a square $S$ containing at least $\frac{n}{k^2-1}$ points of $P$. Let $Q=P\cap S$. Then we have
\[ \phi(Q) \ge \frac{\frac{n}{k^2-1}}{\frac{B(P)}{k^2}}= \frac{k^2}{k^2-1}\cdot\frac{n}{B(P)}=(1+\frac{1}{k^2-1})\delta.\]

This finishes the proof. \qed
\end{proof}

\begin{lemma} \label{lem:denseonecrossing}
    Let $P$ be a generic $\alpha$-dense set of $n> k(2)(101^2-1)^{\log_{1+\frac{1}{101^2-1}}(2\alpha^2)}$ points in the plane, for some $\alpha>1$ and $k(2)$ as in Corollary~\ref{cor:densitybounded}. Then, there is a coloring of $P$ such that $\crossAB(R,B)\ge 1$.
\end{lemma}

\begin{proof}
    Since $P$ is $\alpha$-dense, we know by Lemma~\ref{lem:densesetsaredense} that $\phi(P)\ge \frac{1}{\lceil\alpha \rceil^2}$.
    We repeatedly apply Lemma~\ref{lem:densityimproving} to $P$ with $k=101$. If, in some iteration, we obtain a point set that fills in a $101 \times 101$ grid, we apply Lemma~\ref{lem:fullchessboard} on this set and obtain a crossing. Note that in the coloring from Lemma~\ref{lem:fullchessboard}, all points inside the set except for two are colored blue.  We color the rest of $P$ blue as well.

    It remains to argue that it is not possible
to apply Lemma~\ref{lem:densityimproving} to $P$, with $k=101$, $\log_{1 + \frac{1}{101^2-1}}(3\alpha^2)$ times and never obtain a point set that fills in a $101 \times 101$ grid. Suppose, for the sake of contradiction, that it is possible. We would obtain  a point set $Q$ with $\phi(Q)>3$ and $|Q|>cn$, for $c = (101^2-1)^{-\log_{1 + \frac{1}{101^2-1}}(3\alpha^2)}$. However, this is  impossible if $n>k(3)/c$, where $k(3)$ is the value from Corollary~\ref{cor:densitybounded}. \qed
\end{proof}

We are now ready to prove Theorem~\ref{thm-dense}. 

\dense*
\begin{proof}
    Since $P$ is $\alpha$-dense, it is contained in a square $Q$ of side-length $\alpha\sqrt{n}$. We define $r = k(2) (101^2-1)^{\log_{1 + \frac{1}{101^2-1}}(2\alpha^2)}$  and $x=\sqrt{r} \alpha^2$. 
    We next split $Q$ into $\frac{n\alpha^2}{x^2}= \frac{n}{r\alpha^2}$ cells of side-length $x$. By a standard argument for dense point sets (see Lemmas 12 and 13 in \cite{DumitrescuPach}, for example), we get that, if $n$ is sufficiently large, there are $\Omega(n)$ cells with at least $r$ points inside each of them; we call these cells \emph{rich}. We choose a subset of rich cells in such a way that if a particular cell is chosen, none of the 24 cells closest to it will be chosen. In this way, we can again choose $\Omega(n)$ rich cells. Finally, for each rich cell  $C$, $|P\cap C| \ge r$ and the diameter of $|P\cap C|$ is at most $\sqrt{2r}\alpha^2$.  Thus, $P\cap C$ is $\sqrt{2}\alpha^2$-dense and we can apply Lemma~\ref{lem:denseonecrossing} to obtain a crossing inside each of the chosen rich cells. Note that, in the obtained coloring, in each such cell exactly two points are red. Finally, we color the rest of the set $P$ blue. This finishes the proof.  \qed
\end{proof}

\section{Concluding remarks}

To the best of our knowledge, we were the first to study the bicolored MST crossing number of a set of points (under that or any other name) since Kano, Merino and Uruttia \cite{KanoMU05}. There are many questions still left to explore about this parameter, we mention some of the most interesting ones here. 

The first and most important problem, in our opinion, is improving the lower bound for generic point sets. Currently, the best lower bound is $1$, given by Theorem~\ref{thm-genlowertwocols}. This is very far from the linear upper bound given by Theorem~\ref{thm-genupperbound}. In fact, we believe that Theorem~\ref{thm-genupperbound} is asymptotically tight.  

\begin{conjecture} \label{conj: general}
    Let $P$ be a generic set of $n$ points in the plane. Then $\cross(P) \in \Omega(n)$. 
\end{conjecture}

We remark that proving any reasonable lower bounds for $\cross(P)$ seems to require some deeper insights into the behavior of the MST of $P$ (and subsets of $P$); any significant improvement in this direction would be very interesting. A natural approach to consider is the concept of \textit{bipartite coloring}, as introduced in~\cite{JMS24}. This involves considering the MST of the point set and coloring the points with two colors such that every edge of MST connects points of different colors. Intuitively, bipartite colorings tend to produce well-distributed colorings across the point set. However, this method does not guarantee an interesting lower bound on the number of crossings between the two resulting MSTs. In fact, the \textit{triangular chain} example from \cite{JMS24} shows that the two MSTs in a bipartite coloring can have no crossings at all.
An easier problem in the same direction would be to prove a linear lower bound when we are allowed to discard points first, that is, to prove a lower bound for $\max\{\cross(Q):Q\subseteq P\}$.

In the other direction, it would be interesting to improve the constant from Lemma~\ref{lem-constantlymanycrossings} and determine its exact value. We remark that this constant is at least $6$, as is witnessed by the configuration in Figure~\ref{fig:5edgescrossingshort}.  

\begin{figure}[h]
    \centering
    \includegraphics[scale = 0.7,page=1]{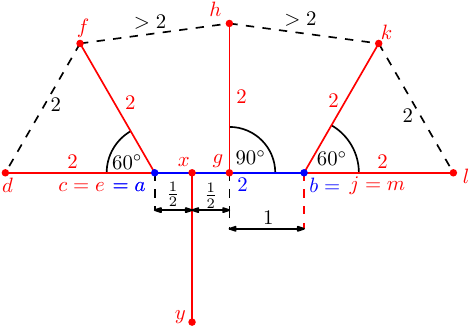}    \caption{Configuration in which $5$ red edges cross a shorter blue edge. To realize this as a generic set we need to perturb slightly. After perturbing, $ab$ should be slightly shorter than all red segments; $xy$ should cross $ab$ with the longer part of $xy$ lying below $ab$, while the other five red segments will cross $ab$ and have their longer parts lying above $ab$. }
    \label{fig:5edgescrossingshort}
\end{figure}

Another interesting open problem is to determine the exact minimum value of $\cross(P)$ when $P$ is generic and in convex position. Currently, Theorem~\ref{thm-convex} and Lemma~\ref{lem-n2-grid-tight-2} tell us that this number lies somewhere in between $\frac{n}{2}$ and $\frac{5n}{8}$. We are unsure which (if any) of these bounds is tight. However, we do remark that we believe that the lower bound of $\frac{n}{2}$ would be the easier one to improve. 

For the sake of clarity, in most of our theorems, we assumed that the point set is generic. However, most of them also hold if we drop this assumption (but still require general position); it is only needed to handle some technical details of the proofs with more care. The only exceptions are Theorems~\ref{thm-grids} (which does not hold in the non-generic setting, see Lemma \ref{lem-nonUniqueMST}), and Theorem~\ref{thm-dense}.
Our proof of Theorem~\ref{thm-dense} heavily relies on the assumption that the minimum spanning trees of the point set and all of its subsets are unique. The main bottleneck is Lemma~\ref{lem:fullchessboard}, the proof of which does not hold without this assumption. It is natural to ask if Theorem~\ref{thm-dense} can be generalized to non-generic point sets using a different approach. 

From a complexity-theoretic point of view, it would be interesting to determine whether the bicolored MST crossing number of a given point set can be computed in polynomial time. We conjecture that this is not possible.
\begin{conjecture}
    Finding the bicolored MST crossing number of a point set in the plane is NP-hard.
\end{conjecture}

\paragraph*{Acknowledgments.} 
T. {Anti\'c}, M. Saumell, F. Schr\"oder and P. Valtr were supported by grant no. 23-04949X of the Czech Science Foundation (GA\v{C}R). T.{Anti\'c} was additionally supported by  GAUK
grant SVV–2025–26082  J. Tkadlec was supported by GAČR project 25-17221S and Charles Univ.\ projects UNCE 24/SCI/008, PRIMUS 24/SCI/012.
%TO DO IN FINAL VERSION: add acknowledgments

\bibliography{bibliography}

@article{KanoMU05,
  author       = {Mikio Kano and
                  Criel Merino and
                  Jorge Urrutia},
  title        = {On plane spanning trees and cycles of multicolored point sets with
                  few intersections},
  journal      = {Information Processing Letters},
  volume       = {93},
  number       = {6},
  pages        = {301--306},
  year         = {2005},
  doi          = {10.1016/J.IPL.2004.12.003},
  bibsource    = {dblp computer science bibliography, https://dblp.org}
}

@article{Valtr1992,
  title = {Convex independent sets and 7-holes in restricted planar point sets},
  volume = {7},
  ISSN = {1432-0444},
  url = {},
  DOI = {10.1007/bf02187831},
  number = {2},
  journal = {Discrete \& Computational Geometry},
  publisher = {Springer Science and Business Media LLC},
  author = {Valtr,  Pavel},
  year = {1992},
  month = feb,
  pages = {135–152}
}

@article{Binnewies18,
    author = {Mikhail Binnewies et. al.},
    title = {Understanding the tumor immune microenvironment (TIME) for effective therapy},
    journal = {Nature Medicine},
    volume = {24},
    year = 2018,
    pages = {541-550}
}

@article{Hei12,
    author = {J.L. Maître and H. Berthoumieux and S.F. Krens and G. Salbreux and F. Jülicher and E. Paluch and C.P. Heisenberg},
    title = {Adhesion functions in cell sorting by mechanically coupling the cortices of adhering cells},
    journal = {Science},
    volume = {338},
    number = {6104},
    year = {2012},
    pages = {253–256}
}

@article{DumitrescuPach,
author = {Dumitrescu, Adrian and Pach, J\'{a}nos and T\'{o}th, G\'{e}za},
title = {Two Trees Are Better than One},
journal = {SIAM Journal on Discrete Mathematics},
volume = {39},
number = {3},
pages = {1883-1893},
year = {2025},
doi = {10.1137/24M1685365},
}

@article{PFRT22,
    author = {G.\ Palla and D. S.\ Fischer and A.\ Regev and F. J.\ Theis.},
    title = {Spatial components of molecular tissue biology},
    journal = {Nature Biotechnology},
    volume = {40},
    year = 2022,
    pages = {308-318}
}

@InProceedings{CDES24,
  author =	{Cultrera di Montesano, Sebastiano and Draganov, Ond\v{r}ej and Edelsbrunner, Herbert and Saghafian, Morteza},
  title =	{{The Euclidean MST-Ratio for Bi-Colored Lattices}},
  booktitle =	{ GD 2024},
  pages =	{3:1--3:23},
  series =	{LIPIcs},
  ISBN =	{978-3-95977-343-0},
  ISSN =	{1868-8969},
  year =	{2024},
  volume =	{320},
  publisher =	{Schloss Dagstuhl -- Leibniz-Zentrum f{\"u}r Informatik},
  address =	{Dagstuhl, Germany},
  URL =		{},
  URN =		{urn:nbn:de:0030-drops-212878},
  doi =		{10.4230/LIPIcs.GD.2024.3}
}

@misc{JMS24,
      title={{On the MST-ratio: Theoretical Bounds and Complexity of Finding the Maximum}}, 
      author={Afrouz Jabal Ameli and Faezeh Motiei and Morteza Saghafian},
      year={2025},
      eprint={2409.11079},
      archivePrefix={arXiv},
      primaryClass={cs.CG},
      url={}, 
}

@inbook{Erds2009,
  title = {A Combinatorial Problem in Geometry},
  ISBN = {9780817648428},
  url = {},
  DOI = {10.1007/978-0-8176-4842-8_3},
  booktitle = {Classic Papers in Combinatorics},
  publisher = {Birkh\"{a}user Boston},
  author = {Erd\"{o}s,  P. and Szckeres,  G.},
  year = {2009},
  pages = {49–56}
}

@misc{FriedmanURL,
  title= {Circles Covering Circles},
author={Erich Friedman},
  url = {https://erich-friedman.github.io/packing/circovcir/},
year={2019}
}

@article{kano2021discrete,
  title={Discrete geometry on colored point sets in the plane—a survey},
  author={Kano, Mikio and Urrutia, Jorge},
  journal={Graphs and Combinatorics},
  volume={37},
  number={1},
  pages={1--53},
  year={2021},
  publisher={Springer}
}

@article{Edelsbrunner1997,
  title = {Cutting dense point sets in half},
  volume = {17},
  ISSN = {1432-0444},
  DOI = {10.1007/pl00009291},
  number = {3},
  journal = {Discrete \& Computational Geometry},
  publisher = {Springer Science and Business Media LLC},
  author = {Edelsbrunner,  Herbert and Valtr,  Pavel and Welzl,  Emo},
  year = {1997},
  month = apr,
  pages = {243–255}
}

@article{Alon1989,
  title = {The maximum size of a convex polygon in a restricted set of points in the plane},
  volume = {4},
  ISSN = {1432-0444},
 
  DOI = {10.1007/bf02187725},
  number = {3},
  journal = {Discrete \& Computational Geometry},
  publisher = {Springer Science and Business Media LLC},
  author = {Alon,  Noga and Katchalski,  Meir and Pulleyblank,   William R.},
  year = {1989},
  month = jun,
  pages = {245–251}
}

@article{Bukh2025,
  title = {Convex polytopes in restricted point sets in $\mathbb{R}^d$},
  ISSN = {2517-5599},
  
  DOI = {10.19086/aic.2025.1},
  journal = {Advances in Combinatorics},
  publisher = {Alliance of Diamond Open Access Journals},
  author = {Bukh,  Boris and Zichao,  Dong},
  year = {2025},
  month = jan 
}

@misc{DumiToth2022,
  author = {Dumitrescu,  Adrian and T{\'{o}}th,  Csaba D.},
  title = {Finding Points in Convex Position in Density-Restricted Sets},
  eprint={2205.03437},
  archivePrefix={arXiv},
  year = {2022},
}

@misc{OurArxiv,
  url = {},
  author = {Antić,  Todor and Saghafian,  Morteza and Saumell,  Maria and Schr\"{o}der,  Felix and Tkadlec,  Josef and Valtr,  Pavel},
  title = {How many times can two minimum spanning trees cross?},
  eprint={2601.20060},
  archivePrefix={arXiv},
  year = {2026},
  copyright = {arXiv.org perpetual,  non-exclusive license}
}

@inproceedings{DumiPachGD,
  doi = {10.4230/LIPICS.GD.2024.9},
  author = {Dumitrescu,  Adrian and Pach,  J{\'{a}}nos},
  language = {en},
  title = {Partitioning Complete Geometric Graphs on Dense Point Sets into Plane Subgraphs},
  publisher = {Schloss Dagstuhl – Leibniz-Zentrum f\"{u}r Informatik},
  year = {2024},
  copyright = {Creative Commons Attribution 4.0 International license},
  booktitle =	{GD 2024},
  pages        = {9:1--9:10},
  series =	{LIPIcs},
   year =	{2024},
  volume =	{320},
}

@article{Kovcs2019,
  title = {Dense Point Sets with Many Halving Lines},
  volume = {64},
  ISSN = {1432-0444},

  DOI = {10.1007/s00454-019-00080-3},
  number = {3},
  journal = {Discrete \& Computational Geometry},
  publisher = {Springer Science and Business Media LLC},
  author = {Kov{\'{a}}cs,  Istv{\'{a}}n and T{\'{o}}th,  Géza},
  year = {2019},
  month = mar,
  pages = {965–984}
}

@article{Valtr1996,
  title = {Lines,  line-point incidences and crossing families in dense sets},
  volume = {16},
  ISSN = {1439-6912},
  DOI = {10.1007/bf01844852},
  number = {2},
  journal = {Combinatorica},
  publisher = {Springer Science and Business Media LLC},
  author = {Valtr,  Pavel},
  year = {1996},
  month = jun,
  pages = {269–294}
}

@phdthesis{Valtr1994,
    author      = {Valtr, Pavel},
    title       = {Planar point sets with bounded ratios of distances},
    type        = {PhD thesis},
    school = {Free University Berlin},
    year        = {1994},
}

\end{document}